\documentclass[times,sort&compress,3p]{elsarticle}
\usepackage[labelfont=bf]{caption}

\usepackage{color}
\usepackage{natbib}
\usepackage{url}
\usepackage{array, booktabs, multirow}
\usepackage{epsfig, graphicx}
\usepackage{fancyheadings}
\usepackage{rotating}
\usepackage[colorlinks,citecolor=blue,urlcolor=blue,pdfpagelabels]{hyperref}
\usepackage{sectsty, secdot}
\usepackage{amsmath, amsfonts, amssymb, amsthm}
\usepackage{bm}
\usepackage[ruled,vlined,linesnumbered]{algorithm2e}

\theoremstyle{plain}
\newtheorem{theorem}{Theorem}

\newtheorem{proposition}{Proposition}
\newtheorem{lemma}{Lemma}
\newtheorem{corollary}{Corollary}

\theoremstyle{definition}

\newtheorem{remark}{Remark}

\newtheorem{assumption}{Assumption}

\makeatletter
\renewenvironment{proof}[1][\proofname]{
	\par
	\pushQED{\qed}
	\normalfont \topsep6pt \@plus6pt \relax
	\trivlist
	\item[\hskip\labelsep \textbf{\textup{#1}}]\ignorespaces
}{
	\popQED\endtrivlist\@endpefalse
}
\makeatother

\newcommand{\subhead}[1]{\noindent\upshape\textbf{#1}}

\newcommand{\bi}{\bm{i}}
\newcommand{\bj}{\bm{j}}
\newcommand{\bt}{\bm{t}}
\newcommand{\bu}{\bm{u}}
\newcommand{\bx}{\bm{x}}
\newcommand{\by}{\bm{y}}
\newcommand{\bz}{\bm{z}}
\newcommand{\bA}{\bm{A}}
\newcommand{\bD}{\bm{D}}
\newcommand{\bI}{\bm{I}}
\newcommand{\bM}{\bm{M}}
\newcommand{\bR}{\bm{R}}

\newcommand{\btau}{\bm{\tau}}
\newcommand{\bmu}{\bm{\mu}}
\newcommand{\bSig}{\bm{\Sigma}}
\newcommand{\tr}{\mathrm{tr}}
\newcommand{\diag}{\mathrm{diag}}

\newcommand{\Cov}{\mathrm{Cov}}
\newcommand{\Var}{\mathrm{Var}}
\newcommand{\E}{\mathrm{E}}
\newcommand{\bbP}{\mathrm{Pr}}
\newcommand{\dto}{\stackrel{d}{\longrightarrow}}
\newcommand{\pto}{\stackrel{P}{\longrightarrow}}
\newcommand{\Fr}{\mathcal{F}}

\providecommand{\triangleq}{\stackrel{\scriptstyle\triangle}{=}}

\begin{document}


\begin{frontmatter}

\title{{Tests for white noise via asymptotically independent U-statistics in high-dimensions}}

\author[1]{Yuanya Xu\corref{mycorrespondingauthor}}

\address[1]{School of Statistics and Data Science, Shanghai University of Finance and Economics, Shanghai 200433, China}

\cortext[mycorrespondingauthor]{Corresponding author. Email address: \url{xu.yuanya@stu.sufe.edu.cn}}

\begin{abstract}
Existing tests for high-dimensional white noise are tailored to the two ends of the sparsity spectrum: maximum-type statistics excel against a few dominant correlations and sum-type statistics against many weak ones, leaving the intermediate regime without a systematic tool. We propose a family of $U$-statistics built from the entry-wise $\ell_a$-norms of the sample autocovariance matrices, in which the even order $a$ interpolates between the sum-type ($a=2$) and maximum-type ($a\to\infty$) endpoints. As $p$ and $T$ diverge jointly, martingale difference theory yields joint asymptotic normality across orders together with asymptotic independence, so that several orders can be combined after a simple calibration. Neither cross-sectional independence nor a mixing condition is required, only a spectral condition on $\bSig_0$; the derivation links cross-sectional correlation to a graph structure that prices each moment edge by edge. Simulations confirm accurate size and power that is stable across sparse, intermediate and dense alternatives, and show that the low-order members of the family retain their level under conditional heteroscedasticity, where the existing tests distort severely.
\end{abstract}


\end{frontmatter}


\section{Introduction}\label{sec:intro}

Testing for white noise is fundamental to diagnostic checking after regression and time series fitting, and modern applications, such as financial panels, genomic screens, sensor arrays, routinely deliver dimensions $p$ comparable to or larger than the sample size $T$. For a $p$-dimensional zero-mean series $\{\bx_t\}_{t=1}^{T}$, the hypotheses are
\begin{equation}\label{eq:hypo}
H_0:\ \bx_1,\ldots,\bx_T\ \mbox{mutually independent}\quad\mbox{vs.}\quad H_1:\ \mbox{serial dependence exists}.
\end{equation}
Let $\bSig_\tau=\E(\bx_t\bx_{t-\tau}^{\top})=(\sigma_{i,j,\tau})_{1 \le i,j \le p}$, abbreviate $\sigma_{i,j,0}$ as $\sigma_{ij}$, and let $\bR_\tau=\bD^{-1/2}\bSig_\tau\bD^{-1/2}$ with $\bD=\diag(\sigma_{11},\ldots,\sigma_{pp})$; hatted symbols denote the usual sample counterparts, with entries $\hat{\sigma}_{i,j,\tau}$ and $\hat{\rho}_{i,j,\tau}$.

For fixed $p$, the Ljung--Box test (\citet{Ljung1978}) and its multivariate extensions (\citet{Hosking1980,Li1981}) are standard tools, with statistics that are essentially sums of squared Frobenius norms of $\{\widehat{\bR}_\tau\}_{\tau=1}^{q}$. Once $p$ diverges proportionally with $T$, $\widehat{\bR}_\tau$ is no longer consistent for $\bR_\tau$, the limiting null distributions are distorted, and both size control and power are severely compromised.

Existing high-dimensional tests fall into three categories. The first aggregates autocorrelations through a norm. \citet{Chang2017} proposed the maximum-type ($\ell_\infty$) statistic
$$
M_q=\max_{1 \le \tau \le q}\max_{1 \le i,j \le p} \left|\hat{\rho}_{i,j,\tau}\right|,
$$
which excels under sparse alternatives and was later refined by \citet{Tsay2020} and \citet{Chen2025}. At the opposite end, \citet{Li2019} proposed the sum-type ($\ell_2$) statistic
$$
S_q=\sum_{\tau=1}^{q}\left\|\widehat{\bSig}_\tau\right\|_F^2,
$$
which pools all correlations and is powerful against dense alternatives. The second category analyses the eigenvalues of high-dimensional autocovariance matrices by random matrix theory (\citet{Li2018,Bose2020}). The third reduces the problem to a lower-dimensional one, either by vectorizing the outer products $\bx_t\bx_{t-\tau}^{\top}$ into a one-sample mean testing problem (\citet{Wang2020}) or by compressing the series into a univariate one before applying the Ljung--Box test (\citet{Ling2021,Zhang2026,Zhang2025ARCH}).

These methods share a common limitation: power depends on the unknown dependence structure. \citet{Feng2026testing} combined maximum-type and sum-of-squares-type statistics into a max-sum test covering both endpoints, but the intermediate regime of moderately sparse dependence, arguably the most common situation in practice, still lacks a systematic tool.

Addressing this gap, we propose a family of $U$-statistics built from the entry-wise $\ell_a$-norm of $\{\widehat{\bSig}_\tau\}_{\tau=1}^{q}$, which capture serial correlations within and across component series without specifying an alternative model. The even order $a$ bridges the two endpoints: $a=2$ recovers a sum-type statistic, larger $a$ shifts the emphasis toward sparse strong signals, and $a\to\infty$ approaches the maximum type. Our contributions are threefold. First, joint asymptotic normality and asymptotic independence across finitely many orders are established as $p$ and $T$ diverge jointly, so that any combination of orders admits a simple calibration. Second, the theory requires neither cross-sectional independence nor a mixing condition, imposing only a spectral condition on $\bSig_0$. Third, we construct equal-weight, Fisher and Cauchy combinations of several orders, among which the equal-weight combination attains the most accurate size control and is recommended as the default.

The idea of $\ell_a$-type $U$-statistics goes back to \citet{He2021asymptotically}, who described the dependence among components by an $\alpha$-mixing condition; mixing, however, formalizes asymptotic independence along a sequential direction and is not a natural device for the cross-section. The $\ell_2$-type statistic of \citet{Li2019} permits weak factors and sparse correlations under $\|\bSig_0\|_2=O(1)$, but its random matrix argument is confined to $a=2$. We impose only $\||\bSig_0|\|_2=O(1)$: slightly stronger than the condition of \citet{Li2019} but of comparable practical scope, far weaker than cross-sectional or block independence, and free of any requirement on the limiting spectral distribution.

Dropping cross-sectional independence is what brings a graph into the analysis. The time indices entering the statistic are linearly ordered, so all that matters is which time points coincide after translation---an algebraic question, formalized by the free-variable calculus of Lemma~\ref{lem:freevar}. The component indices carry no order: cross-sectional dependence is a relation rather than a sequence, which is precisely why a mixing condition sits uneasily on it. A graph is the natural object for an unordered pairwise relation, and representing the moment expansions as weighted multigraphs turns the order in $p$ of each term into a question about connectivity, with Assumption~\ref{asm:3} pricing each edge; see Remark~\ref{rem:graph}.

We further prove that the tests attain power tending to one under the VMA(1) alternative, and the simulations indicate that the lower-order tests remain valid under second-order uncorrelated conditional heteroscedasticity, whereas the existing competitors can distort severely. We stress that the object here is second-order white noise; higher-order dependence is left to future work. Sections~\ref{sec:setup}--\ref{sec:main} below develop the statistic and its asymptotic theory, Section~\ref{sec:sim} reports the numerical studies, Section~\ref{sec:proofs} provides all technical proofs, and Section~\ref{sec:conc} concludes; all additional simulation results appear in the Appendix.

\section{Basic Setup and Statistic}\label{sec:setup}

\subsection{Notation and asymptotic conventions}

Let $u_{T,p}$ and $v_{T,p}$ denote sequences indexed by $T$ and $p$. The asymptotic notation is defined as follows:
\begin{align*}
u_{T,p}&=o(v_{T,p}) &&\mbox{if }\limsup_{T,p\to\infty}|u_{T,p}/v_{T,p}|=0,\\
u_{T,p}&=O(v_{T,p}) &&\mbox{if }\limsup_{T,p\to\infty}|u_{T,p}/v_{T,p}|<\infty,\\
u_{T,p}&=\Theta(v_{T,p}) &&\mbox{if }0<\liminf_{T,p\to\infty}|u_{T,p}/v_{T,p}|\le\limsup_{T,p\to\infty}|u_{T,p}/v_{T,p}|<\infty,\\
u_{T,p}&\sim v_{T,p} &&\mbox{if }\lim_{T,p\to\infty}u_{T,p}/v_{T,p}=1.
\end{align*}

\subsection{Higher-order moments and free variable sets}

For a $p$-dimensional random vector $\bx_t$ with mean $\bmu_t=(\mu_{1,t},\ldots,\mu_{p,t})^{\top}$, the $k$-th order joint central moment at indices $j_1,\ldots,j_k\in\{1,\ldots,p\}$ is
\begin{equation}\label{eq:moment}
\Pi_{j_1,\ldots,j_k,t}=\E\left[\prod_{\ell=1}^{k}\left(x_{j_\ell,t}-\mu_{j_\ell,t}\right)\right].
\end{equation}
For convenience, we denote $\Pi_{j_1,\ldots,j_k,0}$ by $\Pi_{j_1,\ldots,j_k}$.

For a set $S$ consisting of variables, two variables $x,y \in S$ are said to be linearly dependent if their difference is a fixed constant, that is, $x-y=C$ for some constant $C$. The set of free variables $\Fr(S)$ is defined as a maximal subset of $S$ satisfying:
\begin{enumerate}
\item $\Fr(S)$ contains no constant elements;
\item for any distinct $u,v\in\Fr(S)$, $u-v$ is not a fixed constant;
\item for any $w \in S$, there exists $u\in\Fr(S)$ such that $w-u$ is a fixed constant.
\end{enumerate}
The cardinality of $\Fr(S)$ is denoted by $|S|_{F}$. The following lemma records the additivity of translated free variable sets, a property used repeatedly below; its proof, together with those of all theorems, is given in the supplementary material.

\begin{lemma}\label{lem:freevar}
Let $A=\{a_1,\ldots,a_n\}$ be a set of variables and set $b_i=a_i-c_i$ for fixed constants $c_i \ne 0$, writing $B=\{b_1,\ldots,b_n\}$. Then
\[
\Fr(B)=\{a_i-c_i:\ a_i\in\Fr(A)\}
\]
is a free variable set of $B$, and
\[
|\Fr(A)\cup\Fr(B)|=|A|_{F}+|B|_{F}.
\]
\end{lemma}

\subsection{Construction of the test statistic}

Under the null hypothesis, the autocovariance matrices at all positive lags vanish. Motivated by this property, we construct the following $U$-statistic, which aggregates the autocovariance information up to a finite lag $q$:
$$
U_q(a) = \frac{1}{N_q(T,a)} \sum_{\bt_a\in\mathcal{C}_q(T,a)} \sum_{\tau=1}^{q} \sum_{i,j=1}^{p} \prod_{k=1}^{a} \bigl(x_{i,t_k}x_{j,t_k-\tau}\bigr),
$$
where $a$ is a finite positive even integer, and $\mathcal{C}_q(T,a)$ denotes the collection of all strictly increasing sequences $\bt_a=(t_1,\ldots,t_a)$ satisfying
\[
t_k\in\{q+1,\ldots,T\},\qquad 
\min_{1 \le k \le a-1}\left(t_{k+1}-t_k\right)>q,
\]
with cardinality $N_q(T,a)\triangleq\binom{T-aq}{a}$, which is positive when $T \ge aq+a$ and equal to zero otherwise. We reject $H_0$ for large values of $U_q(a)$.

\begin{remark}\label{rem:1}
For any $\bt_a=(t_1,\ldots,t_a)\in\mathcal{C}_q(T,a)$, write $S(\bt_a)=\{t_1,\ldots,t_a\}$ and define its $\tau$-extension
\[
S_\tau(\bt_a)=\{t_1-\tau,\ t_1,\ \ldots,\ t_a-\tau,\ t_a\},\qquad 
1 \le \tau \le q,
\]
for which $|S(\bt_a)|_{F}=|S_\tau(\bt_a)|_{F}=a$ and $|S_\tau(\bt_a)|=2a$. The mapping $S(\bt_a) \mapsto S_\tau(\bt_a)$ is a bijection, so distinct pairs $(\bt_a,\tau) \ne (\bt_a',\tau')$ correspond to distinct extended sets $S_\tau(\bt_a) \ne S_{\tau'}(\bt_a')$.

This one-to-one correspondence makes the summands of $U_q(a)$ mutually orthogonal in expectation. Under strict white noise, if two product terms involve different multisets of time indices, independence factors the expectation of their product across the individual time points; at least one time point then appears in only one of the two terms, and $\E[x_{i,t}]=0$ forces the whole expectation to vanish. This property underlies the moment computations in the proofs.
\end{remark}


\section{Main Results}\label{sec:main}

\subsection{Asymptotic distribution under the null hypothesis}

We derive the asymptotic distribution of $U_q(a)$ under the null hypothesis based on the following assumptions.

\begin{assumption}\label{asm:1}
The sequence $\{\bz_t\}_{1 \le t \le T}$ comprises independent real-valued $p \times 1$ random vectors $\bz_t=(z_{it})_{1\le i\le p}$ with independent components, where each $z_{it}$ satisfies
\[
\E[z_{it}]=0,\qquad \E[z_{it}^2]=1,\qquad \E[z_{it}^4]<\infty.
\]
\end{assumption}

\begin{assumption}\label{asm:2}
Under $H_0$, for any $j_1,j_2,j_3,j_4\in\{1,\ldots,p\}$, the moments defined in \eqref{eq:moment} satisfy $\Pi_{j_1,j_2,j_3}=0$ and
\begin{equation}\label{eq:fourth}
\left|\Pi_{j_1,j_2,j_3,j_4}\right|
\le\kappa\left(
\left|\sigma_{j_1j_2}\sigma_{j_3j_4}\right|
+\left|\sigma_{j_1j_3}\sigma_{j_2j_4}\right|
+\left|\sigma_{j_1j_4}\sigma_{j_2j_3}\right|\right)
\end{equation}
for some constant $\kappa<\infty$ independent of $p$ and of the indices. Condition \eqref{eq:fourth} asks only that the fourth-order moments be dominated by the second-order ones; when $\bSig_0=\bI_p$, for instance, it holds for every innovation distribution with a finite fourth moment, with $\kappa=\max\{1,\E[z_{it}^4]/3\}$.
\end{assumption}

\begin{assumption}\label{asm:3}
Let $\bSig_0=\Cov(\bx_t,\bx_t)=(\sigma_{ij})_{1 \le i,j \le p}\in\mathbb{R}^{p \times p}$ be the covariance matrix of $\bx_t$. There exists a constant $C>0$, independent of $p$, such that
\begin{align}
&\||\bSig_0|\|_2 \le C,\label{eq:spec}\\
&\bSig_0\ \mbox{contains at least }\Theta(p)\ \mbox{entries of constant order}.\label{eq:sparse}
\end{align}
\end{assumption}

\begin{remark}\label{rem:3}
In \eqref{eq:spec}, $\|\cdot\|_2$ denotes the spectral norm and $|\bSig_0|$ is the entry-wise absolute value matrix. By elementary properties of matrix norms, condition \eqref{eq:spec} directly implies the following entry-wise bounds, which will be used repeatedly in the proofs:
\begin{equation}\label{eq:entry}
\max_{1 \le i \le p}\sum_{j=1}^{p}|\sigma_{ij}| \le C\sqrt{p},\qquad
\max_{1 \le i,j \le p}|\sigma_{ij}| \le C.
\end{equation}
The entry-wise $\ell_a$-norm of a matrix is defined as
\[
\|\bA\|_{\ell_a}=\left(\sum_{i,j}|a_{ij}|^{a}\right)^{1/a}.
\]
Condition \eqref{eq:sparse} is mild: it holds automatically when $\bSig_0$ is non-degenerate, since the diagonal entries $\{\sigma_{ii}\}_{i=1}^{p}$ are variances of constant order and already contribute $\Theta(p)$ entries of constant order. Under Assumption~\ref{asm:3}, we have $\|\bSig_0\|_{\ell_a}^{a}=\Theta(p)$.
\end{remark}

\begin{remark}\label{rem:graph}
Assumption~\ref{asm:3} also prices the terms that the moment expansions produce. Under Assumption~\ref{asm:2}, the expectation of a product $\prod_{k}x_{i,t_k}x_{j,t_k-\tau}$ is bounded by sums of products of covariance entries,
$$
\prod_{\{u,v\}}|\sigma_{u,v}|^{\gamma_{u,v}},
$$
indexed by pairs of coordinates, so that what remains is to determine the order in $p$ of such a sum once every coordinate is summed over $\{1,\ldots,p\}$. This is the content of Proposition~\ref{prop:graph}, stated there for an arbitrary non-negative symmetric weight matrix because its proof uses only the spectral and row-sum bounds \eqref{eq:spec} and \eqref{eq:entry}. The bound is proved by reading the exponent array as a weighted multigraph, which turns the
question into one about connectivity. It involves the statistic only through an upper bound on the degrees, so a single statement covers every order $a$ at once.
\end{remark}

\begin{theorem}[Joint asymptotic normality]\label{thm:1}
Under the null hypothesis $H_0$ in \eqref{eq:hypo} and Assumptions~\ref{asm:1}--\ref{asm:3}, suppose the observed data are generated by
\[
\bx_t=\bSig_0^{1/2}\bz_t,\qquad t=1,\ldots,T.
\]
Then, for any finite collection of distinct positive even integers $\{a_1,\ldots,a_m\}$, as $T,p\to\infty$ jointly, the following convergence holds,
\[
\left[\frac{U_q(a_1)}{\sigma(a_1)},\ \ldots,\ \frac{U_q(a_m)}{\sigma(a_m)}\right]^{\top} \dto N(\bm{0},\bI_m),
\]
where the asymptotic variance is $\sigma^2(a)=q\,N_q(T,a)^{-1}\cdot\|\bSig_0\|_{\ell_a}^{2a}$.
\end{theorem}

To implement the test in Theorem~\ref{thm:1}, we introduce the following moment estimator $\hat{\sigma}(a)$ of $\sigma(a)$. The next theorem guarantees that the asymptotic distribution in Theorem~\ref{thm:1} remains valid when $\sigma(a)$ is replaced by $\hat{\sigma}(a)$.

\begin{theorem}[Consistency of the variance estimation]\label{thm:2}
Under the null hypothesis $H_0$ in \eqref{eq:hypo} and Assumptions~\ref{asm:1}--\ref{asm:3}, the estimator
\begin{equation}\label{eq:hatsig}
\hat{\sigma}(a)=\sqrt{q}\,N_q(T,a)^{-3/2}\sum_{\bt_a\in\mathcal{C}_q(T,a)}\sum_{i,j=1}^{p}\prod_{k=1}^{a}x_{i,t_k}x_{j,t_k}
\end{equation}
converges in probability to $\sigma(a)$, that is, $\hat{\sigma}(a)\pto\sigma(a)$, as $T,p\to\infty$.
\end{theorem}

\subsection{Power analysis under alternatives}

\begin{theorem}[Power consistency under the VMA(1) alternative]\label{thm:3}
Consider the first-order vector moving average alternative
\begin{equation}\label{eq:alt}
\bx_t=\bz_t+\bz_{t-1},
\end{equation}
for which the lag-0 and lag-1 covariance matrices are $\bSig_0=2\bI_p$ and $\bSig_1=\bI_p$, respectively. Under the alternative \eqref{eq:alt} and Assumption~\ref{asm:1}, as $T,p\to\infty$,
\[
\frac{U_q(a)}{\sigma(a)}\pto\infty.
\]
Moreover $\hat{\sigma}(a)/\sigma(a)\pto1$, so the same divergence holds for the studentized statistic $U_q(a)/\hat{\sigma}(a)$. Consequently, for any significance level $\alpha\in(0,1)$, the power of the test tends to one:
\[
\lim_{T,p\to\infty}\bbP\left(\frac{U_q(a)}{\hat{\sigma}(a)}>z_{1-\alpha}\right)=1,
\]
where $z_{1-\alpha}$ is the $(1-\alpha)$-quantile of the standard normal distribution.
\end{theorem}

Consistency of a joint test built from several orders follows at once, and requires no independence among $\{U_q(a_k)\}_{k=1}^{K}$: Boole's inequality gives 
$$
\bbP\left\{\bigcup_{k}\left(U_q(a_k)\le M_k\right)\right\}
\le \sum_{k}\bbP\{U_q(a_k)\le M_k\}
\to 0
$$ 
for any finite $\bM$, so $(U_q(a_1),\ldots,U_q(a_K))\pto(\infty,\ldots,\infty)$, and the continuous mapping theorem applies to any non-decreasing aggregation $G$ with $G(\bx)\to\infty$ whenever $\min_k x_k\to\infty$.

\begin{remark}\label{rem:power}
Theorem~\ref{thm:3} gives consistency but not an explicit power function, and the obstruction is structural rather than computational. Under $H_0$ the two families of indices decouple: $\bx_t$ is independent across $t$, and the coordinates enter only through products of the entries of $\bSig_0$, so a temporal matching count and a spatial graph weight can be bounded separately. Under a general VMA(1) alternative $\bx_t=\bA_0\bz_t+\bA_1\bz_{t-1}$, of which \eqref{eq:alt} is the special case $\bA_0=\bA_1=\bI_p$, a covariance factor is indexed by a pair of coordinates together with a lag, being an entry of $\bSig_0$ or of $\bSig_1$ according to the temporal configuration, and the two parts can no longer be controlled in isolation. The same obstruction already affects consistency, not merely the power function: the proof of Theorem~\ref{thm:3} factorizes the relevant expectations across coordinates, which is legitimate only because $\bA_0=\bA_1=\bI_p$ renders the components of $\bx_t$ independent, and this factorization fails for general $\bA_0,\bA_1$. 

Extending the graph of Remark~\ref{rem:graph} so that its edges carry lag labels, and pricing them by $|\bSig_0|$ and $|\bSig_\tau|$ according to the label, appears to be the natural route both to consistency under a general VMA($q$) alternative and to a local power function for every even $a$; we leave this to future work.
\end{remark}

\subsection{Computation of the statistics}

We consider the computation of the statistic $U_q^{(\ell)}(a)$ under the pairwise spacing constraint, where $\ell$ may be expressed as a tuple $(i,j,\tau)$. Define
\[
U_q^{(\ell)}(a)
=\sum_{\bt_a\in\mathcal{C}_q(T,a)} \prod_{k=1}^{a} s_{\ell,t_k},
\]
where $s_{\ell,t_k}$ is a random variable depending on $\ell$ and $t_k$. The formulation is modular: different choices of $s_{\ell,t_k}$ accommodate a variety of statistics. For $\hat{\sigma}(a)$ in \eqref{eq:hatsig}, for instance, it suffices to take
\[
s_{\ell,t}=x_{i,t}x_{j,t},\qquad \ell\in\{(i,j):1 \le i,j \le p\}.
\]

To improve computational efficiency, we propose Algorithm~\ref{alg:1}. Define the state variable $dp(k,t)$ as the accumulated product sum of all valid sequences of length $k$ ending at time $t$, where a sequence is valid if $|t-r|>q$ for every preceding index $r$. Based on this definition, the transition equation is
\[
dp(k,t)=s_{\ell,t}\cdot\sum_{r=1}^{t-q-1}dp(k-1,r).
\]
Introducing the prefix sums
\[
\mathrm{prefix}(k-1,t)=\sum_{r=1}^{t}dp(k-1,r),
\]
the transition equation simplifies to
\[
dp(k,t)=s_{\ell,t}\cdot\mathrm{prefix}(k-1,t-q-1).
\]
This reduces the complexity of Algorithm~\ref{alg:1} from $O(aT^2)$ to $O(aT)$.

\begin{algorithm}[t!]
\renewcommand{\baselinestretch}{1}\normalsize
\caption{Computing the constrained-index sum}\label{alg:1}
\KwIn{the sequence \texttt{s\_list}, integers $a$ and $q$}
\KwOut{$U_q^{(\ell)}(a)$}
$T\leftarrow\mathrm{length}(\texttt{s\_list})$\;
initialize \texttt{dp\_prev} and \texttt{prefix\_prev} as zero vectors of length $T$\;
\For{$t \leftarrow q+1$ \KwTo $T$}{
$\texttt{dp\_prev}(t)\leftarrow\texttt{s\_list}(t)$\;
$\texttt{prefix\_prev}(t)\leftarrow\texttt{prefix\_prev}(t-1)+\texttt{dp\_prev}(t)$\;
}
\For{$k \leftarrow 2$ \KwTo $a$}{
initialize \texttt{dp\_curr} and \texttt{prefix\_curr} as zero vectors of length $T$\;
\For{$t \leftarrow k(q+1)$ \KwTo $T$}{
$r_{\max} \leftarrow t-q-1$\;
\If{$r_{\max} \ge q+1$}{
$\texttt{dp\_curr}(t)\leftarrow\texttt{s\_list}(t)\times\texttt{prefix\_prev}(r_{\max})$\;}
$\texttt{prefix\_curr}(t)\leftarrow\texttt{prefix\_curr}(t-1)+\texttt{dp\_curr}(t)$\;
}
$\texttt{dp\_prev}\leftarrow\texttt{dp\_curr}$\;
$\texttt{prefix\_prev}\leftarrow\texttt{prefix\_curr}$\;}
\Return $\texttt{prefix\_prev}(T)$\;
\end{algorithm}


\section{Numerical Studies}\label{sec:sim}

\subsection{Overall design of the experiments}\label{sec:sim_design}

To assess the finite-sample performance of $U_a$ ($a\in\{2,4,6\}$) and of their combinations, we design experiments covering different innovation distributions, dimensions, signal structures, and noise types. Throughout, the lag order is $q=1$, the nominal level is $\alpha=5\%$, and $U_a$ denotes the studentized statistic $U_q(a)/\hat{\sigma}(a)$. All results are based on $2{,}000$ independent replications, except for the GARCH$(0.1,0.85)$ case of Section~\ref{sec:sim_weak}, which uses $5{,}000$.

In addition to $U_2$, $U_4$, and $U_6$, we examine three combination tests built on their asymptotic independence. The first is the standardized equal-weight sum
$$
U_{\rm adp}=\frac{1}{\sqrt{m}}\sum_{a\in\mathcal{A}}U_a,\qquad
\mathcal{A}=\{2,4,6\},\qquad m=|\mathcal{A}|=3,
$$
where the asymptotic normality and independence of the $U_a$ imply $U_{\rm adp} \dto N(0,1)$ under $H_0$. The Fisher test $U_{\rm F}$ of \citet{Fisher1973} and the Cauchy test $U_{\rm C}$ of \citet{Liu2020} first convert each $U_a$ into its $p$-value via the standard normal distribution function, then combine these $p$-values. Competitors include the maximum-type test $M_q$ of \citet{Chang2017} (critical values by parametric bootstrap), the sum-type test $S_q$ of \citet{Li2019}, and the max-sum test $F_q$ of \citet{Feng2026testing}.

All data are generated from $\bx_t=\bSig_0^{1/2}\bz_t$, where the components of $\bz_t=(z_{1t},\ldots,z_{pt})^{\top}$ are independent univariate series. Depending on the purpose of the experiment, $\bz_t$ is specified as follows:
\begin{enumerate}
\item empirical size: $z_{it} \sim N(0,1)$ (the Gaussian baseline); $z_{it} \sim t(8)/\sqrt{4/3}$, standardized to unit variance; or $z_{it} \sim \mathrm{Gamma}(4,0.5)-2$, centered to zero mean and unit variance;
\item empirical power: $\bz_t$ is generated from VAR(1) or VMA(1) models with standard Gaussian innovations;
\item robustness: $z_{it}$ is a weak white noise process of ARCH, GARCH, or bilinear type.
\end{enumerate}
The contemporaneous covariance matrix $\bSig_0$ takes two structures:
\begin{enumerate}
\item the identity $\bSig_0=\bI_p$, with cross-sectionally uncorrelated components;
\item the AR(1)-type decay $\bSig_0=\bigl(0.5^{|j_1-j_2|}\bigr)_{1 \le j_1,j_2 \le p}$, with cross-sectional correlations decaying exponentially in the index distance.
\end{enumerate}
The dimension ratio is $p/T\in\{0.5,1.0,1.5\}$ and the sample size is $T\in\{100,200,400,800\}$ ($T\in\{100,200,400\}$ in the power study), covering the three regimes $p<T$, $p=T$, and $p>T$ with dimension up to $p=1200$; see Tables~\ref{tab:size_a}--\ref{tab:wn_i} for the detailed $(p,T)$ configurations.

\subsection{Empirical size}\label{sec:sim_size}

To examine size control under the null, the components of $\bz_t$ are i.i.d.\ so that $\{\bx_t\}$ is independent white noise. Table~\ref{tab:size_a} reports the empirical sizes under model (a), in which $z_{it} \sim N(0,1)$; the results under models (b) and (c), in which $z_{it} \sim t(8)/\sqrt{4/3}$ and $z_{it}\sim\mathrm{Gamma}(4,0.5)-2$, are deferred to Tables~\ref{tab:size_b} and \ref{tab:size_c}. In all of these tables, the upper and lower panels correspond respectively to $\bSig_0=\bI_p$ and $\bSig_0=\bigl(0.5^{|j_1-j_2|}\bigr)$.

Table \ref{tab:size_a} shows that $U_2$, $U_4$, and $U_6$ stay close to the nominal level for all $(p,T)$, even when $p/T=1.5$. Among the combination tests, $U_{\rm adp}$ controls size most accurately, while the deviations of $U_{\rm F}$ and $U_{\rm C}$ are confined to small samples and vanish as $T$ grows. This is expected: $U_{\rm adp}$, being a linear combination, relies only on the first two moments of the $U_a$ and their pairwise asymptotic uncorrelatedness, whereas $U_{\rm F}$ and $U_{\rm C}$ map each $U_a$ to a marginal $p$-value via the standard normal CDF and then apply $-2\ln p$ and $\tan\{(0.5-p)\pi\}$, both steep near zero and thus magnifying tail-approximation errors. 

The competitors behave differently from one another. The size of $S_q$ is comparable to that of the $U_a$; $M_q$, which relies on bootstrap calibration, fluctuates more; and $F_q$ is systematically conservative, increasingly so as $p/T$ grows and less so as $T$ grows. Its power is therefore understated in the uncorrected comparisons that follow.

Models (b) and (c) share the same excess kurtosis (1.5) but differ in symmetry: the standardized $t(8)$ innovation in (b) is symmetric, whereas the centered Gamma innovation in (c) has unit skewness. Comparing Tables~\ref{tab:size_b} and \ref{tab:size_c} therefore isolates the effect of asymmetry. Model (b) satisfies the third-order condition of Assumption~\ref{asm:2}; model (c) violates $\Pi_{j_1,j_2,j_3}=0$ through its nonzero third moments. Third-order moments can enter the expansions only via time-coincidence patterns in which some time point carries an odd number of factors; such patterns are heavily restricted by the spacing constraint defining $\mathcal{C}_q(T,a)$, so their contributions are of lower order. Under (c), the empirical sizes of the proposed statistics are indeed close to the nominal level and nearly indistinguishable from those under (b), suggesting that the third-moment condition simplifies the proofs rather than limits the test's applicability.

The competitors respond differently to the skewness. The sizes of $S_q$ stay close to the nominal level, and $F_q$ remains conservative as under model~(a). The parametric bootstrap calibration of $M_q$, however, deteriorates under skewed innovations: its sizes are substantially inflated when $\bSig_0=\bI_p$ and show no improvement as $T$ grows, whereas the inflation is much milder under the correlated design and absent under model~(a). Skewed innovations therefore affect the proposed tests considerably less than the maximum-type competitor, even though model~(c) lies outside the scope of Assumption~\ref{asm:2}.

\begin{table}[t!]
\centering
\caption{Empirical sizes (in percentage) of the tests under model (a), in which $z_{it} \sim N(0,1)$, at the $5\%$ nominal level with $q=1$. The upper and lower panels correspond to $\bSig_0=\bI_p$ and $\bSig_0=\bigl(0.5^{|j_1-j_2|}\bigr)_{1 \le j_1,j_2 \le p}$, respectively. Results are based on 2,000 Monte Carlo replications.}
\label{tab:size_a}
\begin{tabular}{rrrrrrrrrrrr} \hline
\multicolumn{1}{c}{$p$} & \multicolumn{1}{c}{$T$} & \multicolumn{1}{c}{$p/T$} & \multicolumn{1}{c}{$U_2$} & \multicolumn{1}{c}{$U_4$} & \multicolumn{1}{c}{$U_6$} & \multicolumn{1}{c}{$U_{\rm adp}$} & \multicolumn{1}{c}{$U_{\rm F}$} & \multicolumn{1}{c}{$U_{\rm C}$} & \multicolumn{1}{c}{$M_q$} & \multicolumn{1}{c}{$S_q$} & \multicolumn{1}{c}{$F_q$} \\ \hline
\multicolumn{12}{c}{$\bSig_0=\bI_p$}\\ \hline
50 & 100 & 0.5 & 6.05 & 5.80 & 4.70 & 5.60 & 6.35 & 7.00 & 4.15 & 5.70 & 3.80 \\
100 & 200 & 0.5 & 4.85 & 5.80 & 5.40 & 6.05 & 6.60 & 6.65 & 4.95 & 5.10 & 4.65 \\
200 & 400 & 0.5 & 4.70 & 5.70 & 4.70 & 5.15 & 5.15 & 5.00 & 4.55 & 4.85 & 3.65 \\
400 & 800 & 0.5 & 5.20 & 5.30 & 5.60 & 5.05 & 5.40 & 5.65 & 5.60 & 5.15 & 5.00 \\
100 & 100 & 1.0 & 5.60 & 5.35 & 5.05 & 4.85 & 5.60 & 5.55 & 5.00 & 5.50 & 2.80 \\
200 & 200 & 1.0 & 4.95 & 4.95 & 6.20 & 5.95 & 6.35 & 6.30 & 7.20 & 4.95 & 3.45 \\
400 & 400 & 1.0 & 4.25 & 5.00 & 5.70 & 4.85 & 4.85 & 5.20 & 4.85 & 3.85 & 3.45 \\
800 & 800 & 1.0 & 4.85 & 3.95 & 5.60 & 5.05 & 4.80 & 4.50 & 4.80 & 5.00 & 4.10 \\
150 & 100 & 1.5 & 4.25 & 5.35 & 5.50 & 6.25 & 6.35 & 5.30 & 4.35 & 4.40 & 2.45 \\
300 & 200 & 1.5 & 5.05 & 5.65 & 5.85 & 4.70 & 5.15 & 5.85 & 5.80 & 4.90 & 3.35 \\
600 & 400 & 1.5 & 6.25 & 4.05 & 4.65 & 5.05 & 4.85 & 4.20 & 4.85 & 6.30 & 3.40 \\
1200 & 800 & 1.5 & 5.55 & 4.50 & 4.95 & 5.20 & 5.65 & 4.85 & 4.20 & 5.45 & 3.95 \\ \hline
\multicolumn{12}{c}{$\bSig_0=\bigl(0.5^{|j_1-j_2|}\bigr)_{1 \le j_1,j_2 \le p}$}\\ \hline
50 & 100 & 0.5 & 4.55 & 5.90 & 5.50 & 6.15 & 6.35 & 6.35 & 5.30 & 4.35 & 3.95 \\
100 & 200 & 0.5 & 6.15 & 4.90 & 5.25 & 6.05 & 6.80 & 6.80 & 4.90 & 6.35 & 5.15 \\
200 & 400 & 0.5 & 5.55 & 5.95 & 5.05 & 6.20 & 6.45 & 6.40 & 3.50 & 5.05 & 4.60 \\
400 & 800 & 0.5 & 4.70 & 5.10 & 4.45 & 5.50 & 5.05 & 3.95 & 4.90 & 4.50 & 4.05 \\
100 & 100 & 1.0 & 5.10 & 5.10 & 5.60 & 5.70 & 6.20 & 6.05 & 4.15 & 4.95 & 2.90 \\
200 & 200 & 1.0 & 6.20 & 5.55 & 4.90 & 5.70 & 6.20 & 6.35 & 6.20 & 6.10 & 3.90 \\
400 & 400 & 1.0 & 4.50 & 4.95 & 6.15 & 5.70 & 6.05 & 5.95 & 5.85 & 4.55 & 3.80 \\
800 & 800 & 1.0 & 5.15 & 4.75 & 6.20 & 6.40 & 6.00 & 5.55 & 3.65 & 4.75 & 4.20 \\
150 & 100 & 1.5 & 4.50 & 4.95 & 6.15 & 5.30 & 5.20 & 5.35 & 4.00 & 4.25 & 2.60 \\
300 & 200 & 1.5 & 5.35 & 5.05 & 4.70 & 5.40 & 4.65 & 4.85 & 5.20 & 5.20 & 2.75 \\
600 & 400 & 1.5 & 5.50 & 4.70 & 4.55 & 4.90 & 4.95 & 5.15 & 4.70 & 5.30 & 4.25 \\
1200 & 800 & 1.5 & 5.00 & 5.35 & 5.00 & 5.05 & 5.30 & 5.15 & 4.35 & 4.90 & 4.30 \\ \hline
\end{tabular}
\end{table}

\subsection{Empirical power}\label{sec:sim_power}

To gauge the ability of each test to detect serial correlation, we consider two classes of alternatives:
\begin{align}
\mbox{VAR(1)}:&\quad \by_t=\bA\by_{t-1}+\bm{\varepsilon}_t,\qquad \bx_t=\bSig_0^{1/2}\by_t,\label{eq:var1}\\
\mbox{VMA(1)}:&\quad \bm{w}_t=\bm{\varepsilon}_t+\bA\bm{\varepsilon}_{t-1},\qquad \bx_t=\bSig_0^{1/2}\bm{w}_t,\label{eq:vma1}
\end{align}
where $\bm{\varepsilon}_t\stackrel{i.i.d.}{\sim}N(\bm{0},\bI_p)$, and the VAR processes are preceded by a burn-in of 50 periods to remove initialization effects. The coefficient matrix takes the block form
$$
\bA=\theta\begin{pmatrix}
\bI_d & \bm{0}\\
\bm{0} & \bm{0}
\end{pmatrix},\qquad \theta=0.2,
$$
so that only the first $d$ components carry first-order autocorrelation and the remaining $p-d$ stay under the null. Fixing $\theta=0.2$ and varying $d$ traces the sparsity of the signal: sparse ($d=\lfloor0.05p\rfloor$), intermediate ($d=\lfloor0.30p\rfloor$), and dense ($d=\lfloor0.95p\rfloor$). The per-component signal strength is thus held constant while the total signal volume grows, which is exactly what separates sum-type from maximum-type statistics.

Tables~\ref{tab:pow_var_sparse}--\ref{tab:pow_var_dense} report the empirical power under the VAR(1) model with sparse, intermediate, and dense signals; the results under the VMA(1) model are given in Tables~\ref{tab:pow_vma_sparse}--\ref{tab:pow_vma_dense}. In all of these tables, the upper and lower panels again correspond to the two covariance structures.

Relative performance depends heavily on signal sparsity; no single statistic dominates. Under sparse signals $U_6$ is best while $U_2$ misses the few strong correlations; under dense signals the ranking is nearly reversed; under intermediate signals $U_4$ is best. This reflects the $\ell_a$-norm: higher orders stress extreme entries, lower orders global aggregation. The competitors show the same type dependence: $S_q$ and $F_q$ are very powerful under dense signals but deteriorate sharply under sparse ones (especially $S_q$), whereas $M_q$ is adequate under sparse signals and weakest under dense ones. 

In contrast, the three combination tests match or nearly match the best single statistic in every scenario and never collapse---precisely why they are constructed. Against the highest power attained at each configuration, their average and maximum gaps are small, whereas those of $U_2$, $M_q$, and $S_q$ are far larger, amounting to near-complete failure in extreme cases; the combination tests thus sacrifice a little power in the most favourable scenarios in exchange for protection against severe loss in the least favourable ones. The combination tests perform comparably; given the more accurate size control of $U_{\rm adp}$ reported in Section~\ref{sec:sim_size}, we recommend it as the default.

\begin{table}[t!]
\centering
\caption{Empirical power (in percentage) of the tests under the VAR(1) alternative \eqref{eq:var1} with sparse signals ($d=\lfloor0.05p\rfloor$) at the $5\%$ nominal level with $q=1$ and $\theta=0.2$. The upper and lower panels correspond to $\bSig_0=\bI_p$ and $\bSig_0=\bigl(0.5^{|j_1-j_2|}\bigr)_{1 \le j_1,j_2 \le p}$, respectively. Results are based on 2,000 Monte Carlo replications.}
\label{tab:pow_var_sparse}
\begin{tabular}{rrrrrrrrrrrr}
\hline
\multicolumn{1}{c}{$p$} & \multicolumn{1}{c}{$T$} & \multicolumn{1}{c}{$p/T$} & \multicolumn{1}{c}{$U_2$} & \multicolumn{1}{c}{$U_4$} & \multicolumn{1}{c}{$U_6$} & \multicolumn{1}{c}{$U_{\rm adp}$} & \multicolumn{1}{c}{$U_{\rm F}$} & \multicolumn{1}{c}{$U_{\rm C}$} & \multicolumn{1}{c}{$M_q$} & \multicolumn{1}{c}{$S_q$} & \multicolumn{1}{c}{$F_q$} \\ \hline
\multicolumn{12}{c}{$\bSig_0=\bI_p$}\\ \hline
50 & 100 & 0.5 & 6.45 & 9.00 & 5.85 & 8.65 & 8.80 & 8.45 & 6.15 & 6.00 & 4.35 \\
100 & 200 & 0.5 & 9.25 & 21.70 & 23.60 & 29.45 & 29.35 & 26.25 & 21.55 & 9.50 & 15.40 \\
200 & 400 & 0.5 & 13.90 & 80.70 & 94.25 & 95.95 & 96.65 & 95.50 & 88.40 & 14.80 & 79.20 \\
100 & 100 & 1.0 & 6.85 & 8.25 & 7.85 & 10.00 & 9.95 & 9.90 & 7.20 & 7.05 & 4.55 \\
200 & 200 & 1.0 & 9.10 & 19.35 & 27.45 & 32.25 & 32.30 & 29.15 & 20.25 & 10.05 & 14.90 \\
400 & 400 & 1.0 & 15.30 & 85.65 & 98.45 & 98.65 & 99.00 & 98.80 & 93.35 & 17.35 & 85.55 \\
150 & 100 & 1.5 & 6.20 & 6.85 & 7.25 & 8.80 & 8.85 & 8.35 & 5.30 & 6.75 & 3.95 \\
300 & 200 & 1.5 & 7.65 & 19.85 & 28.05 & 31.30 & 30.90 & 28.75 & 16.30 & 8.35 & 12.10 \\
600 & 400 & 1.5 & 15.20 & 86.15 & 99.50 & 99.50 & 99.60 & 99.40 & 95.15 & 17.50 & 88.80 \\ \hline
\multicolumn{12}{c}{$\bSig_0=\bigl(0.5^{|j_1-j_2|}\bigr)_{1 \le j_1,j_2 \le p}$}\\ \hline
50 & 100 & 0.5 & 7.55 & 7.20 & 5.50 & 8.20 & 8.65 & 9.00 & 4.95 & 7.30 & 5.55 \\
100 & 200 & 0.5 & 9.80 & 18.95 & 22.70 & 27.30 & 27.40 & 25.40 & 19.75 & 9.55 & 16.75 \\
200 & 400 & 0.5 & 14.40 & 76.95 & 91.45 & 92.70 & 93.05 & 91.50 & 85.50 & 15.35 & 76.85 \\
100 & 100 & 1.0 & 7.65 & 8.15 & 7.75 & 11.50 & 11.70 & 10.50 & 5.75 & 7.65 & 5.35 \\
200 & 200 & 1.0 & 9.80 & 20.65 & 26.05 & 32.00 & 31.65 & 29.25 & 19.50 & 10.35 & 15.20 \\
400 & 400 & 1.0 & 15.00 & 81.15 & 97.60 & 97.95 & 98.00 & 97.60 & 92.55 & 15.70 & 82.40 \\
150 & 100 & 1.5 & 6.30 & 8.20 & 7.80 & 9.00 & 8.65 & 8.55 & 5.50 & 6.00 & 2.70 \\
300 & 200 & 1.5 & 9.15 & 19.85 & 28.85 & 32.65 & 33.05 & 30.20 & 21.05 & 10.05 & 14.35 \\
600 & 400 & 1.5 & 14.30 & 84.70 & 98.75 & 98.85 & 99.15 & 98.85 & 93.40 & 15.45 & 84.40 \\ \hline
\end{tabular}
\end{table}

\begin{table}[t!]
\centering
\caption{Empirical power (in percentage) of the tests under the VAR(1) alternative \eqref{eq:var1} with intermediate signals ($d=\lfloor0.30p\rfloor$) at the $5\%$ nominal level with $q=1$ and $\theta=0.2$. The upper and lower panels correspond to $\bSig_0=\bI_p$ and $\bSig_0=\bigl(0.5^{|j_1-j_2|}\bigr)_{1 \le j_1,j_2 \le p}$, respectively. Results are based on 2,000 Monte Carlo replications.}
\label{tab:pow_var_medium}
\begin{tabular}{rrrrrrrrrrrr} \hline
\multicolumn{1}{c}{$p$} & \multicolumn{1}{c}{$T$} & \multicolumn{1}{c}{$p/T$} & \multicolumn{1}{c}{$U_2$} & \multicolumn{1}{c}{$U_4$} & \multicolumn{1}{c}{$U_6$} & \multicolumn{1}{c}{$U_{\rm adp}$} & \multicolumn{1}{c}{$U_{\rm F}$} & \multicolumn{1}{c}{$U_{\rm C}$} & \multicolumn{1}{c}{$M_q$} & \multicolumn{1}{c}{$S_q$} & \multicolumn{1}{c}{$F_q$} \\ \hline
\multicolumn{12}{c}{$\bSig_0=\bI_p$}\\ \hline
50 & 100 & 0.5 & 22.25 & 26.90 & 18.35 & 34.05 & 35.55 & 33.15 & 17.85 & 31.00 & 27.65 \\
100 & 200 & 0.5 & 55.40 & 91.95 & 82.80 & 96.95 & 97.00 & 94.75 & 62.45 & 73.60 & 84.95 \\
200 & 400 & 0.5 & 97.50 & 100.00 & 100.00 & 100.00 & 100.00 & 100.00 & 100.00 & 99.95 & 100.00 \\
100 & 100 & 1.0 & 24.15 & 27.15 & 19.15 & 37.65 & 37.60 & 34.15 & 13.55 & 42.15 & 32.10 \\
200 & 200 & 1.0 & 55.10 & 96.45 & 92.30 & 99.30 & 99.25 & 98.45 & 66.55 & 86.95 & 91.75 \\
400 & 400 & 1.0 & 97.40 & 100.00 & 100.00 & 100.00 & 100.00 & 100.00 & 100.00 & 100.00 & 100.00 \\
150 & 100 & 1.5 & 23.15 & 28.95 & 19.40 & 38.35 & 38.25 & 35.00 & 10.65 & 49.35 & 37.60 \\
300 & 200 & 1.5 & 56.80 & 96.95 & 95.00 & 99.80 & 99.70 & 99.35 & 59.15 & 94.50 & 96.70 \\
600 & 400 & 1.5 & 97.95 & 100.00 & 100.00 & 100.00 & 100.00 & 100.00 & 100.00 & 100.00 & 100.00 \\ \hline
\multicolumn{12}{c}{$\bSig_0=\bigl(0.5^{|j_1-j_2|}\bigr)_{1 \le j_1,j_2 \le p}$}\\ \hline
50 & 100 & 0.5 & 25.90 & 28.05 & 16.50 & 34.50 & 36.00 & 32.75 & 15.05 & 31.20 & 28.75 \\
100 & 200 & 0.5 & 55.95 & 89.10 & 79.40 & 94.25 & 94.35 & 92.25 & 66.30 & 65.75 & 80.05 \\
200 & 400 & 0.5 & 96.00 & 100.00 & 100.00 & 100.00 & 100.00 & 100.00 & 100.00 & 98.55 & 100.00 \\
100 & 100 & 1.0 & 24.20 & 28.25 & 18.55 & 37.05 & 37.35 & 32.85 & 14.40 & 33.90 & 27.90 \\
200 & 200 & 1.0 & 55.70 & 93.95 & 89.40 & 98.40 & 98.45 & 97.55 & 67.65 & 76.45 & 86.60 \\
400 & 400 & 1.0 & 96.45 & 100.00 & 100.00 & 100.00 & 100.00 & 100.00 & 100.00 & 99.80 & 100.00 \\
150 & 100 & 1.5 & 21.40 & 28.40 & 18.55 & 38.40 & 37.50 & 32.25 & 11.55 & 36.85 & 27.20 \\
300 & 200 & 1.5 & 54.95 & 95.70 & 93.95 & 98.95 & 98.90 & 98.00 & 66.35 & 84.10 & 90.65 \\
600 & 400 & 1.5 & 97.55 & 100.00 & 100.00 & 100.00 & 100.00 & 100.00 & 100.00 & 99.95 & 100.00 \\ \hline
\end{tabular}
\end{table}

\begin{table}[t!]
\centering
\caption{Empirical power (in percentage) of the tests under the VAR(1) alternative \eqref{eq:var1} with dense signals ($d=\lfloor0.95p\rfloor$) at the $5\%$ nominal level with $q=1$ and $\theta=0.2$. The upper and lower panels correspond to $\bSig_0=\bI_p$ and $\bSig_0=\bigl(0.5^{|j_1-j_2|}\bigr)_{1 \le j_1,j_2 \le p}$, respectively. Results are based on 2,000 Monte Carlo replications.}
\label{tab:pow_var_dense}
\begin{tabular}{rrrrrrrrrrrr} \hline
\multicolumn{1}{c}{$p$} & \multicolumn{1}{c}{$T$} & \multicolumn{1}{c}{$p/T$} & \multicolumn{1}{c}{$U_2$} & \multicolumn{1}{c}{$U_4$} & \multicolumn{1}{c}{$U_6$} & \multicolumn{1}{c}{$U_{\rm adp}$} & \multicolumn{1}{c}{$U_{\rm F}$} & \multicolumn{1}{c}{$U_{\rm C}$} & \multicolumn{1}{c}{$M_q$} & \multicolumn{1}{c}{$S_q$} & \multicolumn{1}{c}{$F_q$} \\ \hline
\multicolumn{12}{c}{$\bSig_0=\bI_p$}\\ \hline
50 & 100 & 0.5 & 81.45 & 71.80 & 37.95 & 85.60 & 90.00 & 88.30 & 31.15 & 99.90 & 99.50 \\
100 & 200 & 0.5 & 99.95 & 100.00 & 99.85 & 100.00 & 100.00 & 100.00 & 94.65 & 100.00 & 100.00 \\
200 & 400 & 0.5 & 100.00 & 100.00 & 100.00 & 100.00 & 100.00 & 100.00 & 100.00 & 100.00 & 100.00 \\
100 & 100 & 1.0 & 85.25 & 76.45 & 43.95 & 90.55 & 94.45 & 92.70 & 30.55 & 100.00 & 100.00 \\
200 & 200 & 1.0 & 99.95 & 100.00 & 100.00 & 100.00 & 100.00 & 100.00 & 96.90 & 100.00 & 100.00 \\
400 & 400 & 1.0 & 100.00 & 100.00 & 100.00 & 100.00 & 100.00 & 100.00 & 100.00 & 100.00 & 100.00 \\
150 & 100 & 1.5 & 88.85 & 77.85 & 45.05 & 94.35 & 96.65 & 94.55 & 31.90 & 100.00 & 100.00 \\
300 & 200 & 1.5 & 100.00 & 100.00 & 100.00 & 100.00 & 100.00 & 100.00 & 97.55 & 100.00 & 100.00 \\
600 & 400 & 1.5 & 100.00 & 100.00 & 100.00 & 100.00 & 100.00 & 100.00 & 100.00 & 100.00 & 100.00 \\ \hline
\multicolumn{12}{c}{$\bSig_0=\bigl(0.5^{|j_1-j_2|}\bigr)_{1 \le j_1,j_2 \le p}$}\\ \hline
50 & 100 & 0.5 & 79.65 & 66.70 & 36.65 & 81.10 & 84.50 & 82.85 & 36.20 & 96.20 & 94.90 \\
100 & 200 & 0.5 & 99.95 & 100.00 & 99.70 & 100.00 & 100.00 & 100.00 & 95.95 & 100.00 & 100.00 \\
200 & 400 & 0.5 & 100.00 & 100.00 & 100.00 & 100.00 & 100.00 & 100.00 & 100.00 & 100.00 & 100.00 \\
100 & 100 & 1.0 & 81.55 & 72.45 & 38.45 & 86.80 & 90.40 & 88.20 & 31.05 & 99.85 & 99.65 \\
200 & 200 & 1.0 & 100.00 & 100.00 & 100.00 & 100.00 & 100.00 & 100.00 & 95.75 & 100.00 & 100.00 \\
400 & 400 & 1.0 & 100.00 & 100.00 & 100.00 & 100.00 & 100.00 & 100.00 & 100.00 & 100.00 & 100.00 \\
150 & 100 & 1.5 & 84.30 & 78.35 & 41.85 & 90.85 & 93.10 & 91.25 & 27.00 & 100.00 & 100.00 \\
300 & 200 & 1.5 & 100.00 & 100.00 & 100.00 & 100.00 & 100.00 & 100.00 & 95.55 & 100.00 & 100.00 \\
600 & 400 & 1.5 & 100.00 & 100.00 & 100.00 & 100.00 & 100.00 & 100.00 & 100.00 & 100.00 & 100.00 \\ \hline
\end{tabular}
\end{table}

\subsection{Robustness under weak white noise}\label{sec:sim_weak}

Our null hypothesis, like that of the competitors, is serial independence, under which all asymptotic distributions are derived. In practice, however, tests for serial correlation are used mainly to detect linear predictability of the conditional mean, and a test should ideally retain its nominal level when the data are conditionally heteroscedastic but serially uncorrelated. This subsection therefore examines the level of each test when independence fails yet serial uncorrelatedness still holds.

We generate $p$ independent univariate weak white noise series and mix them spatially through $\bx_t=\bSig_0^{1/2}\bz_t$. The four processes are:
\begin{enumerate}
\item ARCH(1): $z_{it}=e_{it}\sqrt{h_{it}}$, $h_{it}=0.5+0.3z_{i,t-1}^{2}$, with finite fourth moment;
\item GARCH$(0.1,0.6)$: $z_{it}=e_{it}\sqrt{h_{it}}$, $h_{it}=\omega+\alpha_1z_{i,t-1}^{2}+\beta_1h_{i,t-1}$, with moderate volatility persistence;
\item GARCH$(0.1,0.85)$: same as (ii) with $(\alpha_1,\beta_1)=(0.1,0.85)$ and $\alpha_1+\beta_1=0.95$, close to the boundary of fourth-moment existence;
\item bilinear: $z_{it}=e_{it}+0.5e_{i,t-1}z_{i,t-2}$, which is nonlinearly dependent but outside the conditional-heteroscedasticity family.
\end{enumerate}
Here $e_{it}\stackrel{i.i.d.}{\sim}N(0,1)$; in (ii) and (iii) we set $\omega=1-\alpha_1-\beta_1$ to standardize the unconditional variance to one; and all recursive processes are preceded by a burn-in of 200 periods. All four processes satisfy $\Cov(\bz_t,\bz_{t+k})=\bm{0}$ for $k\ge1$, so second-order uncorrelatedness holds strictly. The sample sizes are $T\in\{200,400,800\}$ and the ratios are $p/T\in\{1.0,1.5\}$. Table~\ref{tab:wn_i} reports the results for $\bSig_0=\bI_p$; those for $\bSig_0=\bigl(0.5^{|j_1-j_2|}\bigr)$ are given in Table~\ref{tab:wn_ii}.

Tables~\ref{tab:wn_i} and~\ref{tab:wn_ii} show that all three competitors suffer substantial size distortion, $M_q$ most severely and $S_q$ and $F_q$ next; the distortion intensifies rather than fades as $T$ grows. The problem is thus asymptotic rather than finite-sample: under conditional heteroscedasticity the limiting distributions of these statistics differ from those derived under independence, so a larger sample only sharpens the mismatch.

Low-order statistics are largely robust to conditional heteroscedasticity: $U_2$ stays near the nominal level under ARCH(1) and both GARCH processes, while $U_4$ inflates only slightly—mainly under the highly persistent GARCH process (iii), where $U_6$ and the combination tests inflate moderately. The contrast with high-order statistics is stark: $U_6$ is dominated by extreme cross-correlations that volatility clustering magnifies, whereas the global averaging in $U_2$ smooths them out. The combination tests inherit part of this bias through $U_6$, with $U_{\rm adp}$ least affected. 

Under the bilinear process, $U_2$, $U_4$, and $U_6$ all inflate moderately, though all of our statistics remain far less distorted than the competitors. In terms of worst-case distortion, $U_2$ and $U_4$ rank smallest, followed by $U_{\rm adp}$, then $U_{\rm F}$ and $U_{\rm C}$, with the competitors largest. When conditional heteroscedasticity or other nonlinear dependence is suspected, $U_2$ or $U_4$ is therefore preferable.

\begin{table}[t!]
\centering
\caption{Empirical sizes (in percentage) of the tests under weak white noise with $\bSig_0=\bI_p$ at the $5\%$ nominal level.}
\label{tab:wn_i}
\begin{tabular}{rrrrrrrrrrrr} \hline
\multicolumn{1}{c}{$p$} & \multicolumn{1}{c}{$T$} & \multicolumn{1}{c}{$p/T$} & \multicolumn{1}{c}{$U_2$} & \multicolumn{1}{c}{$U_4$} & \multicolumn{1}{c}{$U_6$} & \multicolumn{1}{c}{$U_{\rm adp}$} & \multicolumn{1}{c}{$U_{\rm F}$} & \multicolumn{1}{c}{$U_{\rm C}$} & \multicolumn{1}{c}{$M_q$} & \multicolumn{1}{c}{$S_q$} & \multicolumn{1}{c}{$F_q$} \\ \hline
\multicolumn{12}{c}{ARCH(1)}\\ \hline
200 & 200 & 1.0 & 5.32 & 5.48 & 8.00 & 6.26 & 7.58 & 7.30 & 15.08 & 14.18 & 9.72 \\
400 & 400 & 1.0 & 4.94 & 5.62 & 6.52 & 6.10 & 6.32 & 6.04 & 19.32 & 13.50 & 12.94 \\
800 & 800 & 1.0 & 4.98 & 5.44 & 5.98 & 5.68 & 5.76 & 5.60 & 21.90 & 14.94 & 15.38 \\
300 & 200 & 1.5 & 5.02 & 5.58 & 7.44 & 6.40 & 6.98 & 7.62 & 21.70 & 13.92 & 10.60 \\
600 & 400 & 1.5 & 5.34 & 5.40 & 6.64 & 6.16 & 6.54 & 6.34 & 20.08 & 14.84 & 15.12 \\
1200 & 800 & 1.5 & 4.88 & 5.14 & 5.76 & 5.20 & 5.04 & 5.24 & 22.66 & 13.98 & 18.12 \\ \hline
\multicolumn{12}{c}{GARCH$(0.1,0.6)$}\\ \hline
200 & 200 & 1.0 & 5.45 & 5.45 & 5.55 & 5.90 & 6.50 & 6.55 & 6.05 & 7.65 & 3.95 \\
400 & 400 & 1.0 & 4.70 & 5.40 & 5.05 & 5.55 & 5.35 & 5.60 & 4.15 & 6.30 & 2.75 \\
800 & 800 & 1.0 & 4.65 & 5.15 & 5.30 & 5.10 & 5.20 & 5.45 & 6.15 & 6.70 & 4.10 \\
300 & 200 & 1.5 & 4.75 & 5.80 & 6.25 & 5.10 & 5.85 & 5.40 & 4.70 & 6.10 & 2.85 \\
600 & 400 & 1.5 & 5.05 & 5.65 & 5.35 & 5.40 & 5.55 & 5.50 & 5.65 & 6.80 & 4.10 \\
1200 & 800 & 1.5 & 5.45 & 5.15 & 5.20 & 4.85 & 5.65 & 5.95 & 5.20 & 7.20 & 4.75 \\ \hline
\multicolumn{12}{c}{GARCH$(0.1,0.85)$}\\ \hline
200 & 200 & 1.0 & 5.22 & 7.80 & 10.98 & 7.70 & 10.44 & 10.38 & 8.58 & 7.96 & 5.14 \\
400 & 400 & 1.0 & 5.32 & 7.58 & 12.12 & 8.60 & 11.24 & 11.16 & 9.54 & 9.46 & 6.52 \\
800 & 800 & 1.0 & 4.76 & 6.10 & 10.42 & 7.52 & 8.96 & 9.12 & 11.62 & 9.16 & 7.66 \\
300 & 200 & 1.5 & 5.20 & 7.34 & 10.16 & 7.94 & 9.76 & 9.82 & 9.36 & 8.70 & 5.52 \\
600 & 400 & 1.5 & 5.42 & 7.48 & 12.84 & 9.58 & 12.20 & 12.30 & 9.88 & 9.14 & 7.06 \\
1200 & 800 & 1.5 & 4.68 & 5.90 & 11.02 & 8.82 & 10.54 & 9.96 & 11.74 & 9.30 & 7.84 \\ \hline
\multicolumn{12}{c}{Bilinear}\\ \hline
200 & 200 & 1.0 & 8.30 & 7.66 & 8.08 & 10.16 & 10.86 & 10.38 & 21.28 & 15.26 & 14.04 \\
400 & 400 & 1.0 & 8.60 & 6.75 & 7.40 & 10.20 & 11.00 & 10.00 & 22.35 & 15.70 & 16.35 \\
800 & 800 & 1.0 & 8.75 & 7.35 & 6.10 & 9.05 & 9.50 & 8.20 & 17.60 & 15.75 & 16.55 \\
300 & 200 & 1.5 & 9.35 & 6.85 & 7.75 & 9.85 & 10.85 & 10.95 & 23.90 & 15.55 & 14.45 \\
600 & 400 & 1.5 & 8.85 & 6.95 & 7.80 & 9.90 & 10.10 & 9.20 & 22.95 & 15.40 & 17.60 \\
1200 & 800 & 1.5 & 8.65 & 5.45 & 6.40 & 8.20 & 7.75 & 7.70 & 23.40 & 14.90 & 17.45 \\ \hline
\end{tabular}
\end{table}


\section{Technical Proofs}\label{sec:proofs}

\begin{proof}[Proof of Lemma~\ref{lem:freevar}.]
We first verify that $\{a_i-c_i:\ a_i\in\Fr(A)\}$ satisfies the three defining conditions of a free variable set. The translation $a_i \mapsto a_i-c_i$ preserves the dependence relations among the variables: if $a_j$ and $a_k$ are linearly dependent, then
\[
(a_j-c_j)-(a_k-c_k)=(a_j-a_k)-(c_j-c_k)
\]
remains a fixed constant, where $a_j-a_k$ is constant by the dependence and $c_j-c_k$ is the difference of two fixed translation constants. The translation therefore maps the free variables of $\Fr(A)$ to free variables of $B$ and preserves the cardinality, that is,
\[
|\{a_i-c_i: a_i\in\Fr(A)\}|=|\Fr(A)|=|A|_{F},
\]
so we may take $\Fr(B)=\{a_i-c_i: a_i\in\Fr(A)\}$.

By the inclusion--exclusion principle,
\[
|\Fr(A)\cup\Fr(B)|=|\Fr(A)|+|\Fr(B)|-|\Fr(A)\cap\Fr(B)|.
\]
We show that $\Fr(A)\cap\Fr(B)=\emptyset$. Suppose $v\in\Fr(A)\cap\Fr(B)$. Then there exist $a_{i_1},a_{i_2}\in\Fr(A)$ such that
\[
v=a_{i_1}=a_{i_2}-c_{i_2},
\]
whence $a_{i_2}-a_{i_1}=c_{i_2}$, so $a_{i_1}$ and $a_{i_2}$ are linearly dependent. Since $\Fr(A)$ consists of mutually independent free variables, any two distinct elements of $\Fr(A)$ are linearly independent, and thus $a_{i_1}=a_{i_2}$. Substituting back yields $c_{i_2}=0$, contradicting $c_i\ne0$. The intersection is therefore empty, and the claim follows.
\end{proof}

Lemma~\ref{lem:freevar} governs the ordered time indices. The unordered coordinate indices are governed by the following bound, which is a statement about matrices alone: the graph-theoretic argument enters only in its proof. We state it for an arbitrary weight matrix rather than for $|\bSig_0|$, because the proof uses nothing about the entries beyond the three norm bounds in \eqref{eq:Wcond}, and because it is invoked at three separate points, with different degree bounds, in what follows.

Let $\bm{W}=(w_{u,v})_{1 \le u,v \le p}$ be a symmetric matrix with non-negative entries. For an integer $n\ge2$, let $\bm{\gamma}=(\gamma_{u,v})_{1 \le u<v \le n}$ be an array of non-negative integers, set $\gamma_{v,u}=\gamma_{u,v}$, and define
\begin{equation}\label{eq:Wdef}
S(\bm{\gamma};\bm{W})=\sum_{k_1,\ldots,k_n=1}^{p} \prod_{1 \le u<v \le n} w_{k_u,k_v}^{\,\gamma_{u,v}}, \qquad
d_u=\sum_{v \ne u}\gamma_{u,v},
\end{equation}
the outer sum running over all assignments of values in $[p]=\{1,\ldots,p\}$ to the $n$ positions. Call $\bm{\gamma}$ indecomposable if $\{1,\ldots,n\}$ admits no partition into non-empty sets $A$ and $B$ with $\gamma_{u,v}=0$ for every $u \in A$ and $v \in B$. A path of $\bm{\gamma}$ on $L$ positions is a sequence $u_1,\ldots,u_L$ of distinct positions with $\gamma_{u_\ell,u_{\ell+1}}>0$ for $\ell=1,\ldots,L-1$.

\begin{proposition}\label{prop:graph}
Suppose that, for constants $C_1,C_2,C_3$ independent of $p$,
\begin{equation}\label{eq:Wcond}
\|\bm{W}\|_2\le C_1,\qquad
\max_{1 \le u \le p}\ \sum_{v=1}^{p}w_{u,v} \le C_2\sqrt{p},\qquad
\max_{1 \le u,v \le p}w_{u,v} \le C_3,
\end{equation}
and let $\bm{\gamma}$ be indecomposable with $2 \le d_u \le \Delta$ for every $u$, where $\Delta$ is a constant independent of $p$. Then
\begin{itemize}
\item[(i)] $S(\bm{\gamma};\bm{W}) \le C\,p^{\,n/2}$, where $C$ depends only on $n$, $\Delta$ and $C_1,C_2,C_3$;
\item[(ii)] if moreover $\bm{\gamma}$ has a path on $L\ge4$ positions, then $S(\bm{\gamma};\bm{W}) \le C\,p^{\,(n-1)/2} = o\bigl(p^{\,n/2}\bigr)$.
\end{itemize}
\end{proposition}

\begin{proof}[Proof of Proposition~\ref{prop:graph}.]
Consider first $n=2$, in which case $\gamma_{1,2}=d_1\ge2$ and
\[
S(\bm{\gamma};\bm{W})
= \sum_{k_1,k_2=1}^{p}w_{k_1,k_2}^{\,d_1}
\le C_3^{\Delta-2}\sum_{k_1,k_2=1}^{p}w_{k_1,k_2}^{2}
= C_3^{\Delta-2}\|\bm{W}\|_F^{2}
\le C_3^{\Delta-2}C_1^2\,p,
\]
where the first inequality uses the third bound in \eqref{eq:Wcond} and the second uses $\|\bm{A}\|_F\le\sqrt{p}\,\|\bm{A}\|_2$. This is $O(p^{n/2})$ with $n=2$.

Now let $n\ge3$. It is convenient to read $\bm{\gamma}$ as a multigraph $H$ on the vertex set $\{1,\ldots,n\}$, in which $u$ and $v$ are joined by $\gamma_{u,v}$ parallel edges; indecomposability says that $H$ is connected, and $d_u$ is the degree of $u$. Let $T$ be a spanning tree of $H$ and split the product in \eqref{eq:Wdef} into a tree part and a remainder,
\[
\prod_{1 \le u<v \le n}w_{k_u,k_v}^{\,\gamma_{u,v}}
=\prod_{\{u,v\}\in E(T)}w_{k_u,k_v}
\cdot\prod_{1 \le u<v \le n}
w_{k_u,k_v}^{\,\gamma_{u,v}-\bm{1}_{\{u,v\}\in E(T)}}.
\]
Since $w_{u,v} \le C_3$, the remainder is at most $C_3^{M}$ with
\[
M=\sum_{1 \le u<v \le n}\left(\gamma_{u,v}-\bm{1}_{\{u,v\}\in E(T)}\right)
=\frac{1}{2}\sum_{u=1}^{n}d_u-(n-1)
\le\frac{n\Delta}{2}-n+1,
\]
a constant. Hence
\begin{equation}\label{eq:WHWT}
S(\bm{\gamma};\bm{W})
\le C\sum_{k_1,\ldots,k_n=1}^{p}\ \prod_{\{u,v\}\in E(T)}w_{k_u,k_v}.
\end{equation}

It remains to bound the right-hand side of \eqref{eq:WHWT}. Take a diameter path $P=(v_1,\ldots,v_L)$ of $T$ as the backbone, so that $T$ has diameter $L-1$. Two lower bounds on $L$ are available. First, a tree on $n\ge3$ vertices has diameter at least two, whence $L\ge3$. Second, if $H$ has a path on $L_0$ vertices, choose $T$ to contain it; its endpoints then lie at distance $L_0-1$ in $T$, so the diameter path of $T$ has at least $L_0$ vertices. Together, $L \ge \max\{3,\ L_0\}$.

By the maximality of the diameter no branches issue from the endpoints $v_1$ and $v_L$, so deleting the $L-1$ backbone edges leaves $L$ rooted subtrees, the $\ell$-th rooted at $v_\ell$; those rooted at $v_1$ and $v_L$ consist of the root alone. Write $\beta(v_\ell)$ for the number of non-root vertices of the $\ell$-th subtree.

The subtrees are summed out first, one vertex at a time. We claim that for any tree $T'$ rooted at $r$ with $m$ non-root vertices,
\begin{equation}\label{eq:peel}
\max_{k_r\in[p]}\ \sum_{(k_u)_{u \ne r}\in[p]^{m}}\ \prod_{\{u,v\} \in E(T')}w_{k_u,k_v}
\le\left(C_2\sqrt{p}\right)^{m}.
\end{equation}
This holds trivially for $m=0$. For $m\ge1$, pick a leaf $z \ne r$ of $T'$ and let $w$ be its unique neighbour. Summing over $k_z$ first and using the second bound in \eqref{eq:Wcond},
\[
\sum_{k_z=1}^{p}w_{k_w,k_z} \le C_2\sqrt{p}
\qquad\mbox{uniformly in }k_w,
\]
so this factor may be extracted; deleting $z$ leaves a tree with $m-1$ non-root vertices, and \eqref{eq:peel} follows by induction.

Because the $L$ subtrees share no vertex other than their roots, which lie on the backbone, their contributions multiply, and \eqref{eq:peel} applied to each of them, uniformly in the value $k_{v_\ell}$ at its root, contributes the factor
\[
\prod_{\ell=1}^{L}\left(C_2\sqrt{p}\right)^{\beta(v_\ell)}
=\left(C_2\sqrt{p}\right)^{\sum_{\ell=1}^{L}\beta(v_\ell)}.
\]
The exponents therefore add, and the resulting count is
$$
\sum_{\ell=1}^{L}\beta(v_\ell)=n-L,
$$
since the non-root vertices of the $L$ subtrees are precisely the $n-L$ vertices off the backbone, each belonging to exactly one subtree. What remains is the sum over the backbone, a path with $L-1$ edges, so that
$$
\sum_{k_1,\ldots,k_n=1}^{p}\ \prod_{\{u,v\} \in E(T)}w_{k_u,k_v}
\le C_2^{\,n-L}\sqrt{p}^{\,n-L} \sum_{k_{v_1},\ldots,k_{v_L}=1}^{p}\ \prod_{\ell=1}^{L-1}w_{k_{v_\ell},k_{v_{\ell+1}}}\\
=C\sqrt{p}^{\,n-L}\cdot\bm{1}^{\top}\bm{W}^{L-1}\bm{1}.
$$
By the first bound in \eqref{eq:Wcond} and the sub-multiplicativity of the spectral norm,
\begin{equation}\label{eq:chain}
\bm{1}^{\top}\bm{W}^{L-1}\bm{1}
\le p\,\bigl\|\bm{W}^{L-1}\bigr\|_1
\le p\sqrt{p}\,\bigl\|\bm{W}^{L-1}\bigr\|_2
\le p\sqrt{p}\,\|\bm{W}\|_2^{L-1}
\le C_1^{L-1}p\sqrt{p},
\end{equation}
so that
$$
S(\bm{\gamma};\bm{W}) \le C\sqrt{p}^{\,n-L+3}.
$$
Part (i) follows from $L\ge3$, and part (ii) from $L \ge L_0 \ge 4$.
\end{proof}

\begin{corollary}\label{cor:zeta}
Let $\bm{W}$ satisfy \eqref{eq:Wcond} and let $\bm{\gamma}$ be an exponent array on $n=2K$ positions with $2 \le d_u \le \Delta$ for every $u$. Splitting $\{1,\ldots,2K\}$ into the maximal blocks on which $\bm{\gamma}$ is indecomposable, and writing $n_1,\ldots,n_r$ for their sizes and $\bm{\gamma}_1,\ldots,\bm{\gamma}_r$ for the corresponding sub-arrays, the sum factorizes as
\[
S(\bm{\gamma};\bm{W})=\prod_{c=1}^{r}S(\bm{\gamma}_c;\bm{W}),
\]
and $S(\bm{\gamma};\bm{W})=O(p^{K})$. If, in addition, some block has a path on four or more positions, then $S(\bm{\gamma};\bm{W})=o(p^{K})$.
\end{corollary}

\begin{proof}[Proof of Corollary~\ref{cor:zeta}.]
Since $d_u\ge2$, no block is a singleton, so $n_c\ge2$ for every $c$, while $\sum_{c}n_c=2K$. Proposition~\ref{prop:graph}(i) then yields $\prod_{c}S(\bm{\gamma}_c;\bm{W}) \le C\prod_{c}p^{n_c/2} = Cp^{K}$. For the second claim, apply Proposition~\ref{prop:graph}(ii) to the distinguished block and part~(i) to the remaining ones.
\end{proof}

\begin{remark}\label{rem:sigma0}
Throughout the proofs below the weight matrix is $\bm{W}=|\bSig_0|$, the entry-wise absolute value of $\bSig_0$, for which \eqref{eq:Wcond} is exactly the spectral bound \eqref{eq:spec} together with its consequences \eqref{eq:entry}. Nothing else about $\bSig_0$ is used, so Proposition~\ref{prop:graph} applies verbatim to any other weight matrix obeying \eqref{eq:Wcond}; the proof also carries over, with $\bm{W}^{L-1}$ in \eqref{eq:chain} replaced by a product of distinct matrices, when the exponents are labelled and each label carries its own weight matrix.
\end{remark}

\begin{proof}[Proof of Theorem~\ref{thm:1}.]
To establish the joint limiting distribution of the statistics, by the Cram\'er--Wold theorem it suffices to show that any fixed linear combination converges to a normal distribution. For constants $t_1,\ldots,t_m$ with $\sum_{r=1}^{m}t_r^2=1$, define
\begin{equation}\label{eq:ZT}
Z_T=\sum_{r=1}^{m}t_r\frac{U_q(a_r)}{\sigma(a_r)}.
\end{equation}
We prove $Z_T \dto N(0,1)$.

Let $\Fr_0=\{\emptyset,\Omega\}$ and $\Fr_h=\sigma\{\bx_1,\ldots,\bx_h\}$ for $h=1,\ldots,T$, and let $\E_h(\cdot)$ denote the conditional expectation given $\Fr_h$. Define
\begin{equation}\label{eq:D}
D_{T,h}=(\E_h-\E_{h-1})Z_T,\qquad h=1,\ldots,T.
\end{equation}
Since $\E_0(\cdot)=\E(\cdot)$ and $\E(Z_T)=0$, we have $Z_T=\sum_{h=1}^{T}D_{T,h}$, and $\{D_{T,h},\Fr_h\}$ forms a martingale difference sequence. By the martingale central limit theorem (Theorem 35.12 of \cite{Billingsley2012}), $Z_T \dto N(0,1)$ follows once we verify
\begin{equation}\label{eq:clt}
\sum_{h=1}^{T}\E_{h-1}(D_{T,h}^2) \pto 1,\qquad
\sum_{h=1}^{T}\E(D_{T,h}^4) \to 0.
\end{equation}

Write $c(a)=\sigma(a)N_q(T,a)$ and define
\begin{equation}\label{eq:A}
A_{T,h,a}=(\E_h-\E_{h-1})\left\{\frac{U_q(a)}{\sigma(a)}\right\},
\end{equation}
whose explicit form is
\[
A_{T,h,a}=\begin{cases}
0, & h<aq+a,\\[6pt]
\displaystyle\frac{1}{c(a)} \sum_{\bt_{a-1}\in\mathcal{C}_q(h-q-1,a-1)} \sum_{\tau=1}^{q} \sum_{i,j=1}^{p} x_{i,h}x_{j,h-\tau} \prod_{k=1}^{a-1} x_{i,t_k}x_{j,t_k-\tau}, & h \ge aq+a.
\end{cases}
\]
By \eqref{eq:D},
$$
D_{T,h}=\sum_{r=1}^{m}t_rA_{T,h,a_r},
$$
and
\[
\E_{h-1}(D_{T,h}^2)
=\sum_{1 \le r_1,r_2 \le m} t_{r_1}t_{r_2} \E_{h-1}\left(A_{T,h,a_{r_1}}A_{T,h,a_{r_2}}\right).
\]

We verify \eqref{eq:clt} in three steps:
\begin{enumerate}
\item[(i)] the expectation of the conditional variance converges to one,
\begin{equation}\label{eq:p10}
\E\left[\sum_{h=1}^{T}\E_{h-1}(D_{T,h}^2)\right]=1,
\end{equation}
\item[(ii)] the variance of the conditional variance tends to zero,
\begin{equation}\label{eq:p11}
\Var\left[\sum_{h=1}^{T}\E_{h-1}(D_{T,h}^2)\right] \to 0,
\end{equation}
\item[(iii)] the fourth-moment condition,
\begin{equation}\label{eq:p12}
\sum_{h=1}^{T}\E(D_{T,h}^4) \to 0.
\end{equation}
\end{enumerate}
Combining the three steps establishes \eqref{eq:clt}, and the martingale central limit theorem yields $Z_T \dto N(0,1)$, the joint asymptotic normality then follows from the Cram\'er--Wold theorem.

\subhead{Proof of \eqref{eq:p10}.}
We first establish the orthogonality of the $D_{T,h}$: for any $h_1<h_2$, $\E[D_{T,h_1}D_{T,h_2}]=0$. Indeed, $D_{T,h_1}$ is $\Fr_{h_1}$-measurable, and $h_1 \le h_2-1$ implies that $D_{T,h_1}$ is also $\Fr_{h_2-1}$-measurable. Therefore
\[
\E[D_{T,h_1}D_{T,h_2}]
=\E\big[D_{T,h_1}\,\E_{h_2-1}(D_{T,h_2})\big]
=0.
\]
By the orthogonality and $\E(Z_T)=0$,
\begin{equation}\label{eq:vardecomp}
\sum_{h=1}^{T}\E[D_{T,h}^2]
=\E\left[\left(\sum_{h=1}^{T}D_{T,h}\right)^2\right]
=\Var(Z_T).
\end{equation}

We next compute $\Var(Z_T)$. By \eqref{eq:ZT},
\begin{equation}\label{eq:varZT}
\Var(Z_T)
=\sum_{r=1}^{m}t_r^2\frac{\Var[U_q(a_r)]}{\sigma^2(a_r)}
+ \sum_{1 \le r_1 \ne r_2 \le m} t_{r_1}t_{r_2} \frac{\Cov[U_q(a_{r_1}),U_q(a_{r_2})]}{\sigma(a_{r_1})\sigma(a_{r_2})}.
\end{equation}
Consider first $\Var[U_q(a)]$. Under $H_0$, $\E[U_q(a)]=0$, so the variance equals the second moment:
$$
\Var[U_q(a)]
=\frac{1}{N_q(T,a)^2}\sum_{\bt_a,\bu_a\in\mathcal{C}_q(T,a)}
\sum_{\tau_1,\tau_2=1}^{q}\sum_{i_1,j_1,i_2,j_2=1}^{p} \E\left[
\prod_{k=1}^{a}x_{i_1,t_k}x_{j_1,t_k-\tau_1}x_{i_2,u_k}x_{j_2,u_k-\tau_2}
\right].
$$
By Remark~\ref{rem:1}, the expectation on the right-hand side is non-zero only when $\bt_a=\bu_a$ and $\tau_1=\tau_2$, in which case it equals $\sigma_{i_1,i_2}^{a}\sigma_{j_1,j_2}^{a}$. Consequently,
\[
\Var[U_q(a)]
=\frac{q}{N_q(T,a)}\sum_{i_1,i_2,j_1,j_2=1}^{p}\sigma_{i_1,i_2}^{a}\sigma_{j_1,j_2}^{a}
=\sigma^2(a).
\]
Consider next $\Cov[U_q(a_1),U_q(a_2)]$. When $a_1 \ne a_2$, the index tuples $\bt_{a_1}$ (of length $a_1$) and $\bu_{a_2}$ (of length $a_2$) necessarily differ, and Remark~\ref{rem:1} gives $\Cov[U_q(a_1),U_q(a_2)]=0$. Substituting into \eqref{eq:varZT},
\[
\Var(Z_T)=\sum_{r=1}^{m}t_r^2 \cdot 1+0=\sum_{r=1}^{m}t_r^2=1.
\]
Combining this with \eqref{eq:vardecomp} proves \eqref{eq:p10}.

\subhead{Proof of \eqref{eq:p11}.}
By the Cauchy--Schwarz inequality,
\begin{equation}\label{eq:varbound}
\Var\left[\sum_{h=1}^{T}\E_{h-1}(D_{T,h}^2)\right]
\le CT^2\max_{\substack{1 \le h \le T\\ 1 \le r_1,r_2 \le m}}
\Var\left[\E_{h-1}\left(A_{T,h,a_{r_1}}A_{T,h,a_{r_2}}\right)\right].
\end{equation}
It therefore suffices to show that, for any $a_1,a_2$ and $1 \le h \le T$,
\begin{equation}\label{eq:keyvar}
\Var\left[\E_{h-1}\left(A_{T,h,a_1}A_{T,h,a_2}\right)\right]=o(T^{-2}).
\end{equation}

To reduce \eqref{eq:keyvar}, we first compute $\E(A_{T,h,a_1}A_{T,h,a_2})$. Combining \eqref{eq:A} with an argument parallel to the proof of \eqref{eq:p10} yields
\begin{equation}\label{eq:expectA}
\E\left(A_{T,h,a_1}A_{T,h,a_2}\right)=
\begin{cases}
N_q(h-q-1,a-1)/N_q(T,a), & a_1=a_2=a,\\
0, & a_1 \ne a_2.
\end{cases}
\end{equation}
Since $N_q(T,a)=\Theta(T^a)$, we have $\E(A_{T,h,a}^2)=O(T^{-1})$ when $a_1=a_2=a$. Moreover,
\[
\Var\left[\E_{h-1}\left(A_{T,h,a_1}A_{T,h,a_2}\right)\right]
=\E\left[\E_{h-1}\left(A_{T,h,a_1}A_{T,h,a_2}\right)^2\right]
- \left[\E\left(A_{T,h,a_1}A_{T,h,a_2}\right)\right]^2,
\]
so proving \eqref{eq:keyvar} reduces to proving
\begin{equation}\label{eq:secondmoment}
\E\left[\E_{h-1}\left(A_{T,h,a_1}A_{T,h,a_2}\right)^2\right]=
\begin{cases}
O(T^{-2}), & a_1=a_2,\\
o(T^{-2}), & a_1 \ne a_2.
\end{cases}
\end{equation}

Expand the left-hand side of \eqref{eq:secondmoment}. Let
\[
\btau=(\tau_1,\ldots,\tau_4)\in[q]^4,\qquad
\bi=(i_1,\ldots,i_4)\in[p]^4,\qquad
\bj=(j_1,\ldots,j_4)\in[p]^4,
\]
where $[n]=\{1,\ldots,n\}$. Then
\[
\E\left[\E_{h-1}\left(A_{T,h,a_1}A_{T,h,a_2}\right)^2\right]
=\frac{1}{c(a_1)^2c(a_2)^2}
\sum_{\substack{\bt,\tilde{\bt}\in\mathcal{C}_q(h-q-1,a_1-1)\\ \bu,\tilde{\bu}\in\mathcal{C}_q(h-q-1,a_2-1)}}\
\sum_{\btau\in[q]^4}
\E\left[Q\right],
\]
where, writing $\bt=\bt_{a_1-1}$, $\tilde{\bt}=\tilde{\bt}_{a_1-1}$, $\bu=\bu_{a_2-1}$, and $\tilde{\bu}=\tilde{\bu}_{a_2-1}$,
$$
Q(\bt,\tilde{\bt},\bu,\tilde{\bu},\btau)
=\sum_{\bi,\bj\in[p]^4} \sigma_{i_1,i_2}\sigma_{i_3,i_4} \times \prod_{k=1}^{4}x_{j_k,h-\tau_k} \times \prod_{k=1}^{a_1-1}\left(x_{i_1,t_k}x_{j_1,t_k-\tau_1}x_{i_3,\tilde{t}_k}x_{j_3,\tilde{t}_k-\tau_3}\right) \times \prod_{k=1}^{a_2-1}\left(x_{i_2,u_k}x_{j_2,u_k-\tau_2}x_{i_4,\tilde{u}_k}x_{j_4,\tilde{u}_k-\tau_4}\right).
$$

We now characterize the conditions under which $\E[Q(\bt,\tilde{\bt},\bu,\tilde{\bu},\btau)] \ne 0$. Since
\[
\min(h-\tau_1,\ldots,h-\tau_4)>\max(t_{a_1-1},\ldots,\tilde{u}_{a_2-1}),
\]
the factor $\prod_{k=1}^{4}x_{j_k,h-\tau_k}$ involves time indices later than all the others, so its expectation can be taken separately from the remaining factors. This yields two necessary conditions.
\begin{enumerate}
\item[(i)] The components of $\btau$ must appear in pairs:
\begin{equation}\label{eq:taucond}
\forall\,c\in[q],\qquad
\left|\{k\in\{1,2,3,4\}:\tau_k=c\}\right|\in\{0,2,4\}.
\end{equation}
\item[(ii)] Overlap constraint on the time indices. For the multiset
\[
\mathcal{S}=S_{\tau_1}(\bt) \uplus S_{\tau_2}(\bu) \uplus S_{\tau_3}(\tilde{\bt}) \uplus S_{\tau_4}(\tilde{\bu}),
\]
where $\uplus$ denotes the multiset union that preserves multiplicities, every element of $\mathcal{S}$ must have even multiplicity: an element of odd multiplicity would leave an isolated factor $x_{i,t}$, whose expectation vanishes. Since $|\mathcal{S}|=4(a_1+a_2-2)$, the number of distinct elements of $\mathcal{S}$ is at most half of that,
\begin{equation}\label{eq:S1bound}
\left|S_1\right|\le 2(a_1+a_2-2),\qquad
S_1=S_{\tau_1}(\bt) \cup S_{\tau_2}(\bu) \cup S_{\tau_3}(\tilde{\bt}) \cup S_{\tau_4}(\tilde{\bu}).
\end{equation}

Write $S_2=S(\bt) \cup S(\bu) \cup S(\tilde{\bt}) \cup S(\tilde{\bu})$ for the unshifted time points and $S_3=\Delta_{\tau_1}(\bt) \cup \Delta_{\tau_2}(\bu) \cup \Delta_{\tau_3}(\tilde{\bt}) \cup \Delta_{\tau_4}(\tilde{\bu})$ for the shifted ones, where $\Delta_\tau(\bt)=S_\tau(\bt) \setminus S(\bt)$, so that $S_1=S_2 \cup S_3$. We claim that
\begin{equation}\label{eq:boundS2}
\left|S_2\right|_{F} \le a_1+a_2-2.
\end{equation}
Suppose instead that $|S_2|_{F}>a_1+a_2-2$. Each element of $S_3$ is obtained from an element of $S_2$ by translation by one of $-\tau_1,\ldots,-\tau_4$, all of which are non-zero, so Lemma~\ref{lem:freevar} applies with $A=S_2$ and $B=S_3$ and gives
\[
\left|\Fr(S_2)\cup\Fr(S_3)\right|=|S_2|_{F}+|S_3|_{F}=2|S_2|_{F},
\]
the last equality because translation preserves the number of free variables. Since
$\Fr(S_2) \subseteq S_2$ and $\Fr(S_3) \subseteq S_3$, we have
$\Fr(S_2)\cup\Fr(S_3) \subseteq S_1$, whence
\[
\left|S_1\right|\ \ge\ 2|S_2|_{F}\ >\ 2(a_1+a_2-2),
\]
contradicting \eqref{eq:S1bound}. This proves \eqref{eq:boundS2}.
\end{enumerate}

We next show that $\E[Q]=O(p^4)$. The third-order condition of Assumption~\ref{asm:2} forces every element of the multiset $\mathcal{S}\uplus\{h-\tau_1,\ldots,h-\tau_4\}$ to have multiplicity
2 or 4. Under $H_0$ the expectation factorizes across time points; a point of multiplicity two contributes exactly one factor $\sigma_{u,v}$, and a point of multiplicity four is bounded by the fourth-order condition \eqref{eq:fourth}. If exactly $m$ elements have multiplicity 4, where $0 \le m \le 2\min(a_1,a_2)-1$, then $|\E[Q]|$ is at most $\kappa^{m}$ times a sum of $3^{m}$ terms of the form
$$
\zeta=\sum_{\bi,\bj\in[p]^4}\prod_{\{u,v\}\subseteq\{i_1,\ldots,j_4\}}
\left|\sigma_{u,v}\right|^{\gamma_{u,v}}, \qquad
\sum_{v \ne u}\gamma_{u,v}=
\begin{cases}
a_1, & u\in\{i_1,j_1,i_3,j_3\},\\
a_2, & u\in\{i_2,j_2,i_4,j_4\}.
\end{cases}
$$
The leading configuration isolated in Step 1 below has all multiplicities equal to two, that is $m=0$, so the asymptotic variance does not depend on $\kappa$; the constant enters only the terms that are bounded rather than evaluated. It therefore remains to prove that $\zeta=O(p^4)$.

Estimating $\zeta$ is a purely combinatorial problem: $p^8$ terms are summed over the eight indices $(i_1,j_1,\ldots,i_4,j_4)$, each term being a product of factors $|\sigma_{u,v}|$. In the notation of \eqref{eq:Wdef}, $\zeta=S(\bm{\gamma};|\bSig_0|)$ with $n=8$ positions and exponent array $\bm{\gamma}$, whose position degrees satisfy
\[
d_u=\sum_{v \ne u}\gamma_{u,v}\in\{a_1,a_2\},
\]
so that $d_u \ge 2$ for every $u$. Corollary~\ref{cor:zeta}, applied with $K=4$ and $\Delta=\max(a_1,a_2)$, gives $\zeta=O(p^4)$ and hence $\E[Q]=O(p^4)$.

By \eqref{eq:boundS2} the time indices carry at most $a_1+a_2-2$ free variables, each ranging over $\{1,\ldots,h-q-1\}$, so the number of nonzero terms is $O(h^{a_1+a_2-2})$, each contributing $O(p^4)$. Combined with Assumption~\ref{asm:3},
$$
\sigma^2(a)=q\,N_q(T,a)^{-1}\cdot\|\bSig_0\|_{\ell_a}^{2a}
=\Theta(p^2T^{-a}),\qquad
c(a)^2=\sigma^2(a)N_q(T,a)^2=\Theta(p^2T^{a}).
$$
Consequently,
\[
\E\left[\E_{h-1}\left(A_{T,h,a_1}A_{T,h,a_2}\right)^2\right]
=O\!\left(\frac{h^{a_1+a_2-2}\,p^{4}}{c(a_1)^2c(a_2)^2}\right)
=O(T^{-2}),
\]
uniformly in $1 \le h \le T$.

Finally, we complete the proof of \eqref{eq:keyvar} in two steps.

\noindent{\it Step 1: identifying the leading terms.} Here $\bm{\gamma}$ is an exponent array on $n=8$ positions, indexed by the coordinates $(i_1,j_1,\ldots,i_4,j_4)$, and every position has degree $d_u\in\{a_1,a_2\}$, hence $d_u\ge2$. Splitting the eight positions into the maximal blocks on which $\bm{\gamma}$ is indecomposable, as in Corollary~\ref{cor:zeta}, the block sizes $(n_1,\ldots,n_r)$ form a partition of 8 into parts of size at least two.

The $\sigma$ factors in $Q$ appear in pairs: the pairing $\sigma_{i_1,i_2}\cdot\sigma_{i_3,i_4}$ arises from expanding the conditional expectation $\E_{h-1}$, and the pairings among the $j$-indices are generated by the condition \eqref{eq:taucond} on $\btau$. A paired factor contributes a positive exponent, so the two positions of a pair lie in the same block. Every block is therefore a disjoint union of pairs and has even size, which leaves the five profiles listed below. A block of size at least four contains two such pairs, say $\{a,b\}$ and $\{c,d\}$, and connectedness forces an edge between them, say $b$--$c$; then $a,b,c,d$ is a path on four positions, and Proposition~\ref{prop:graph}(ii) gives $S(\bm{\gamma}_c;|\bSig_0|)=o(p^{\,n_c/2})$.

\begin{center}
\renewcommand{\baselinestretch}{1}\normalsize
\begin{tabular}{|c|c|l|} \hline
Block sizes & Bound & Status \\ \hline
$(8)$ & $o(p^4)$ & a block of size $n_c\ge4$\\
$(6,2)$ & $o(p^4)$ & a block of size $n_c\ge4$\\
$(4,4)$ & $o(p^4)$ & a block of size $n_c\ge4$\\
$(4,2,2)$ & $o(p^4)$ & a block of size $n_c\ge4$\\
$(2,2,2,2)$ & $\Theta(p^4)$ & retained\\ \hline
\end{tabular}
\end{center}

Hence $\E[Q]=\Theta(p^4)$ if and only if $\bm{\gamma}$ splits into four blocks of size two. This requires $\bt=\bu$ and $\tilde{\bt}=\tilde{\bu}$, hence $a_1=a_2$, and is realized by the assignment $\tau_1=\tau_2\ne\tau_3=\tau_4$.

\noindent{\it Step 2: computing the left-hand side of \eqref{eq:keyvar}.} By Step 1,
\[
\E\left[\E_{h-1}\left(A_{T,h,a_1}A_{T,h,a_2}\right)^2\right]=
\begin{cases}
\dfrac{a_1^2\, h^{2(a_1-1)}}{T^{2a_1}}+o(T^{-2}), & a_1=a_2,\\
o(T^{-2}), & a_1 \ne a_2.
\end{cases}
\]
Combining this with \eqref{eq:expectA} yields \eqref{eq:keyvar}. Substituting into \eqref{eq:varbound} completes the derivation of \eqref{eq:p11}.

\subhead{Proof of \eqref{eq:p12}.}
By \eqref{eq:A},
\[
\sum_{h=1}^{T}\E(D_{T,h}^{4})
\le C\sum_{h=1}^{T}\sum_{1 \le r_1,\ldots,r_4 \le m}
\E\left(\prod_{k=1}^{4}A_{T,h,a_{r_k}}\right).
\]
To establish \eqref{eq:p12}, it suffices to verify that
\begin{equation}\label{eq:prodA}
\E\left(\prod_{k=1}^{4}A_{T,h,a_k}\right)=O(T^{-2}).
\end{equation}
Expanding via \eqref{eq:A} gives
\[
\E\left(\prod_{k=1}^{4}A_{T,h,a_k}\right)
=\frac{1}{\prod_{k=1}^{4}c(a_k)}
\sum_{\bt^{(1)},\ldots,\bt^{(4)}}\sum_{\btau\in[q]^4}
\E\left[Q^{\star}\left(\bt^{(1)},\bt^{(2)},\bt^{(3)},\bt^{(4)},\btau\right)\right],
\]
where $\bt^{(k)}\in\mathcal{C}_q(h-q-1,a_k-1)$ for $k\in\{1,2,3,4\}$ and
\[
Q^{\star}\left(\bt^{(1)},\bt^{(2)},\bt^{(3)},\bt^{(4)},\btau\right)
=\sum_{\bi,\bj\in[p]^4}\prod_{k=1}^{4}\left(x_{i_k,h}\,x_{j_k,h-\tau_k}\right)
\times\prod_{k=1}^{4}\prod_{\ell=1}^{a_k-1}
\left(x_{i_k,t_{k,\ell}}\,x_{j_k,t_{k,\ell}-\tau_k}\right).
\]
Unlike $Q$, in which the conditional expectation $\E_{h-1}$ has already contracted the $i$-indices into the factor $\sigma_{i_1,i_2}\sigma_{i_3,i_4}$, all eight factors at time $h$
survive in $Q^{\star}$. The remaining analysis nevertheless follows the framework used for \eqref{eq:p11}, with $a_1,a_2$ replaced by $a_1,\ldots,a_4$.
\begin{itemize}
\item Non-vanishing conditions. The inequality $\E[Q^{\star}]\ne0$ again relies on two key constraints: on the one hand, $\btau$ must satisfy \eqref{eq:taucond}; on the other hand, writing $c=(a_1+\cdots+a_4-4)/2$, we must have $\left|\bigcup_{k=1}^{4}S(\bt^{(k)})\right|_F \le c$.
\item Upper bound. The exponent array attached to $Q^{\star}$ again has eight positions, weight matrix $\bm{W}=|\bSig_0|$, and degrees in $\{a_1,\ldots,a_4\}$. Corollary~\ref{cor:zeta} applies verbatim with $K=4$ and $\Delta=\max_{k}a_k$, since Proposition~\ref{prop:graph} depends on the orders only through the degree bound $\Delta$. Hence $\E[Q^{\star}]=O(p^4)$.
\end{itemize}
The number of nonzero terms is therefore $O(T^c)$, each contributing $O(p^4)$, while $\prod_{k=1}^{4}c(a_k)=\Theta(p^4T^{c+2})$. This gives \eqref{eq:prodA}, and hence \eqref{eq:p12}.
\end{proof}

\begin{proof}[Proof of Theorem \ref{thm:2}.]
To prove $r(a)=\hat{\sigma}(a)/\sigma(a)\pto1$, it suffices to verify
\[
\E[r(a)]=1,\qquad \Var[r(a)]=\E[r(a)^2]-1 \to 0.
\]
The first identity follows from $\E[\hat{\sigma}(a)]=\sigma(a)$, so it remains to show $\E[r(a)^2]\to1$. Expanding via \eqref{eq:hatsig},
$$
\E[r(a)^2]
=\frac{1}{N_q(T,a)^2\|\bSig_0\|_{\ell_a}^{2a}}
\sum_{\bt_a,\bu_a\in\mathcal{C}_q(T,a)}\xi(\bt_a,\bu_a),
$$
where the key term is
\[
\xi(\bt_a,\bu_a)=\sum_{i_1,j_1,i_2,j_2=1}^{p}\E\left[
\prod_{k=1}^{a}x_{i_1,t_k}x_{j_1,t_k}x_{i_2,u_k}x_{j_2,u_k}
\right].
\]
Let $s=|S(\bt_a) \cap S(\bu_a)|$ and define
\[
\xi_s(\bt_a,\bu_a)
\triangleq\sum_{i_1,j_1,i_2,j_2=1}^{p}\E\left[
\prod_{k=1}^{a}\left(x_{i_1,t_k}x_{j_1,t_k}x_{i_2,u_k}x_{j_2,u_k}\right)
\cdot \bm{1}_{\{|S(\bt_a) \cap S(\bu_a)|=s\}}
\right],
\]
so that $\xi(\bt_a,\bu_a)=\sum_{s=0}^{a}\xi_s(\bt_a,\bu_a)$; for brevity, we write $\xi_s$ for $\xi_s(\bt_a,\bu_a)$ below.

We control $|\xi_s|$ for each $s$. Under Assumption~\ref{asm:2},
$$
|\xi_s|
\le C\sum_{i_1,j_1,i_2,j_2=1}^{p} \left|\sigma_{i_1,j_1}\right|^{a-s}\left|\sigma_{i_2,j_2}\right|^{a-s}
\times\Bigl(
\left|\sigma_{i_1,j_1}\right|^{s}\left|\sigma_{i_2,j_2}\right|^{s}
+\left|\sigma_{i_1,i_2}\right|^{s}\left|\sigma_{j_1,j_2}\right|^{s}
+\left|\sigma_{i_1,j_2}\right|^{s}\left|\sigma_{j_1,i_2}\right|^{s}
\Bigr)
\triangleq \xi_{s,1}+\xi_{s,2}+\xi_{s,3}.
$$
We estimate the three terms separately.
\begin{enumerate}
\item[(i)] For $s=0$ or $s=a$, the three terms degenerate into direct products of $\sigma$ factors. Taking $s=0$ as an example,
\[
|\xi_0|
=\sum_{i_1,j_1,i_2,j_2=1}^{p}\left|\sigma_{i_1,j_1}\right|^{a}\left|\sigma_{i_2,j_2}\right|^{a}
=\|\bSig_0\|_{\ell_a}^{2a},
\]
and similarly $|\xi_a|=O(\|\bSig_0\|_{\ell_a}^{2a})$.
\item[(ii)] For $1 \le s \le a-1$, first,
\[
\xi_{s,1}
\le C\sum_{i_1,j_1,i_2,j_2=1}^{p}
\left|\sigma_{i_1,j_1}\right|^{a}\left|\sigma_{i_2,j_2}\right|^{a}
=O\left(\|\bSig_0\|_{\ell_a}^{2a}\right).
\]
Second, by $|\sigma_{ij}| \le C$ (Assumption~\ref{asm:3}),
$$
\xi_{s,2}
\le C\sum_{i_1,j_1,i_2,j_2=1}^{p}
\left|\sigma_{i_1,j_1}\right|\left|\sigma_{j_1,j_2}\right| \left|\sigma_{j_2,i_2}\right|\left|\sigma_{i_2,i_1}\right|
\le C\tr\left(|\bSig_0|^4\right)
\le Cp\||\bSig_0|\|_2^4
=O(p),
$$
where the last step uses $\||\bSig_0|\|_2 \le C$ from Assumption~\ref{asm:3}, and $\xi_{s,3}=O(p)$ follows by a symmetric argument.
\end{enumerate}
In summary,
\begin{equation}\label{eq:xibound}
|\xi_s|=O\left(\|\bSig_0\|_{\ell_a}^{2a}\right),\qquad 0 \le s \le a,
\end{equation}
and hence $|\xi(\bt_a,\bu_a)|=O\left(\|\bSig_0\|_{\ell_a}^{2a}\right)$.

It remains to analyze the expectation.
\begin{enumerate}
\item[(i)] Disjoint pairs ($s=0$): when the two sample sequences $S(\bt_a)$ and $S(\bu_a)$ share no time point, there are $\Theta(T^{2a})$ such pairs, each contributing the constant $|\xi_0|=\|\bSig_0\|_{\ell_a}^{2a}$. Hence
\[
\E\left[r(a)^2\cdot\bm{1}_{\{|S(\bt_a) \cap S(\bu_a)|=0\}}\right]
= 1+O(T^{-1}).
\]
\item[(ii)] Intersecting pairs ($s\ge1$): when the two sequences share at least one time point, there are $O(T^{2a-1})$ such pairs, each contributing $O(\|\bSig_0\|_{\ell_a}^{2a})$ by \eqref{eq:xibound}. Hence
\[
\E\left[r(a)^2\cdot\bm{1}_{\{|S(\bt_a) \cap S(\bu_a)|\ge1\}}\right]
= O(T^{-1}) \to 0.
\]
\end{enumerate}
Combining the two cases gives $\E[r(a)^2]=1+O(T^{-1})$, so $\Var[r(a)]\to0$, and Chebyshev's inequality yields $r(a)\pto1$.
\end{proof}

\begin{proof}[Proof of Theorem~\ref{thm:3}.]
Let $Y_{T,p}=U_q(a)/\sigma(a)$. By the definition of divergence in probability to $+\infty$, we must show that, for any fixed $M>0$,
\begin{equation}\label{eq:target}
\lim_{\min(T,p)\to\infty}\bbP\left(Y_{T,p} \le M\right)=0.
\end{equation}

Fix $M>0$. Since $\mu_{T,p}=\E[Y_{T,p}]=\Theta(T^{a/2})\to\infty$ as $\min(T,p)\to\infty$, there exists a positive integer $N$ such that $\mu_{T,p}>M$ whenever $T,p>N$. Setting $\varepsilon_{T,p}=\mu_{T,p}-M>0$, we have
\[
Y_{T,p} \le M \quad\Longrightarrow\quad |Y_{T,p}-\mu_{T,p}|\ge\varepsilon_{T,p},
\]
and Chebyshev's inequality gives
\[
\bbP(Y_{T,p} \le M) \le \frac{\Var(Y_{T,p})}{\varepsilon_{T,p}^2}.
\]
It remains to show that the right-hand side tends to zero. Note that
\[
\varepsilon_{T,p}^2=\mu_{T,p}^2\left(1-\frac{M}{\mu_{T,p}}\right)^2
\sim\mu_{T,p}^2,\qquad \min(T,p)\to\infty.
\]
If we establish the variance bound
\begin{equation}\label{eq:varbound2}
\Var(Y_{T,p}) \le Cp^{-1}T^{a}+C,
\end{equation}
then
\[
\frac{\Var(Y_{T,p})}{\varepsilon_{T,p}^2}
\le\frac{Cp^{-1}T^{a}+C}{\mu_{T,p}^2\left(1-M/\mu_{T,p}\right)^2}
\to0,
\]
using $\mu_{T,p}^2=\Theta(T^a)$. This proves \eqref{eq:target}, that is, $Y_{T,p}\pto\infty$.

Since $\Var(Y_{T,p})=\Var[U_q(a)]/\sigma^2(a)$, we first compute $\Var[U_q(a)]$. Under the alternative \eqref{eq:alt}, $\E[U_q(a)]=p$, so
\[
\Var[U_q(a)]=\E[U_q(a)^2]-p^2,
\]
where
$$
\E[U_q(a)^2]
=\frac{1}{N_q(T,a)^2}\sum_{\bt_a,\bu_a\in\mathcal{C}_q(T,a)}
\sum_{\tau_1,\tau_2=1}^{q}\sum_{i_1,j_1,i_2,j_2=1}^{p} \E\left(
\prod_{k=1}^{a}x_{i_1,t_k}x_{j_1,t_k-\tau_1}x_{i_2,u_k}x_{j_2,u_k-\tau_2}
\right).
$$
Under \eqref{eq:alt}, the dependence structure of $\bx_t$ makes the above expectation non-zero only when the indices satisfy specific pairing relations. The non-vanishing cases fall into exactly four categories.

\noindent{\it Case (i): $i_1=j_1 \ne i_2=j_2$.} After substituting the indices, the $a$ product factors decompose into the product of $\sigma_{i_1,i_1,1}^{a}$ and $\sigma_{i_2,i_2,1}^{a}$, summing over $\bt_a$, $\bu_a$, $\tau_1$, and $\tau_2$ gives $N_q(T,a)^2$, which cancels the normalization. Hence
\[
\E\left[U_q(a)^2\cdot\bm{1}_{\{i_1=j_1 \ne i_2=j_2\}}\right]
=\sum_{1 \le i_1 \ne i_2 \le p} \sigma_{i_1,i_1,1}^{a}\cdot\sigma_{i_2,i_2,1}^{a}
=p(p-1).
\]

\noindent{\it Case (ii): $i_1=i_2 \ne j_1=j_2$.} In this case,
\begin{align*}
\E\left[U_q(a)^2\cdot\bm{1}_{\{i_1=i_2 \ne j_1=j_2\}}\right]
=\frac{1}{N_q(T,a)^2} \sum_{\bt_a,\bu_a\in\mathcal{C}_q(T,a)}
\sum_{\tau_1,\tau_2=1}^{q} \sum_{1 \le i \ne j \le p}
\E\left(\prod_{k=1}^{a}x_{i,t_k}x_{i,u_k}\right)
\E\left(\prod_{k=1}^{a}x_{j,t_k-\tau_1}x_{j,u_k-\tau_2}\right),
\end{align*}
which is non-zero only if, for every $t \in S(\bt_a)$, there exists a unique $u \in S(\bu_a)$ with $|t-u|\le1$, hence $|S(\bt_a) \cup S(\bu_a)|_F=a$, and the number of valid time-index combinations is $O(T^a)$ rather than $O(T^{2a})$. Combining this with $\sum_{1 \le i \ne j \le p}O(1)=O(p^2)$,
\[
\E\left[U_q(a)^2\cdot\bm{1}_{\{i_1=i_2 \ne j_1=j_2\}}\right]
=\frac{O(T^a) \cdot O(p^2)}{\Theta(T^{2a})}
=O(p^2T^{-a}).
\]

\noindent{\it Case (iii): $i_1=j_2 \ne i_2=j_1$.} By an argument parallel to Case (ii),
\[
\E\left[U_q(a)^2\cdot\bm{1}_{\{i_1=j_2 \ne i_2=j_1\}}\right]=O(p^2T^{-a}).
\]

\noindent{\it Case (iv): $i_1=j_1=i_2=j_2$.} The four indices coalesce into a single free index $i$, whose summation contributes $O(p)$. There is no additional constraint on the time indices, so $|S(\bt_a) \cup S(\bu_a)|\le2a$ and the sum over $(\bt_a,\bu_a)$ contributes $O(N_q(T,a)^2)$, which cancels the normalization to $O(1)$; the product expectation is $O(1)$, obtained by substituting \eqref{eq:alt}. Hence
\[
\E\left[U_q(a)^2\cdot\bm{1}_{\{i_1=j_1=i_2=j_2\}}\right]=O(p).
\]

Combining the four cases,
\[
\E[U_q(a)^2]=p(p-1)+O(p)+O(p^2T^{-a}),
\]
so
\[
\Var[U_q(a)]=\E[U_q(a)^2]-p^2=O(p)+O(p^2T^{-a}).
\]
Since $\sigma^2(a)=\Theta(p^2T^{-a})$,
\[
\Var(Y_{T,p})=\frac{\Var[U_q(a)]}{\sigma^2(a)}
=\frac{O(p)+O(p^2T^{-a})}{\Theta(p^2T^{-a})}
=O(p^{-1}T^{a})+O(1),
\]
that is, \eqref{eq:varbound2} holds.

It remains to show that $r(a)=\hat{\sigma}(a)/\sigma(a)$ converges in probability to one under the alternative \eqref{eq:alt}. Two features are used. First, $\hat{\sigma}(a)$ involves only the lag-0 products $x_{i,t_k}x_{j,t_k}$. Second, the components of $\bx_t$ are mutually independent, because $\bA_0=\bA_1=\bI_p$ makes $x_{i,t}=z_{i,t}+z_{i,t-1}$ a function of the $i$-th component of the innovations alone.

Consider first the mean. Within a single tuple $\bt_a\in\mathcal{C}_q(T,a)$ the spacing constraint $t_{k+1}-t_k>q\ge1$ gives $t_{k+1}-t_k\ge2$, so $\bx_{t_1},\ldots,\bx_{t_a}$ depend on disjoint sets of innovations and are mutually independent. Hence
\[
\E\left[\prod_{k=1}^{a}x_{i,t_k}x_{j,t_k}\right]
=\prod_{k=1}^{a}\E\left[x_{i,t_k}x_{j,t_k}\right]=\sigma_{ij}^{a},
\]
so that $\E[\hat{\sigma}(a)]=\sqrt{q}\,N_q(T,a)^{-1/2}\|\bSig_0\|_{\ell_a}^{a}=\sigma(a)$ and $\E[r(a)]=1$.

For the second moment we classify by the coordinate indices, as in the computation of $\Var[U_q(a)]$ above. Writing $\|\bSig_0\|_{\ell_a}^{a}=2^{a}p$,
\begin{align*}
\E[r(a)^2]=\frac{1}{N_q(T,a)^2\,4^{a}p^{2}}
\sum_{\bt_a,\bu_a\in\mathcal{C}_q(T,a)}\ \sum_{i_1,j_1,i_2,j_2=1}^{p}
\E\left[\prod_{k=1}^{a}x_{i_1,t_k}x_{j_1,t_k}x_{i_2,u_k}x_{j_2,u_k}\right].
\end{align*}
Since distinct components are independent with mean zero, the summand vanishes unless the four coordinate indices pair off, which leaves the same four cases as before.

\noindent(i) $i_1=j_1 \ne i_2=j_2$. Independence across components and within each tuple gives
\[
\E\left[\prod_{k=1}^{a}x_{i_1,t_k}^{2}\right]\E\left[\prod_{k=1}^{a}x_{i_2,u_k}^{2}\right]
=\prod_{k=1}^{a}\E\left[x_{i_1,t_k}^{2}\right]\prod_{k=1}^{a}\E\left[x_{i_2,u_k}^{2}\right]
=4^{a},
\]
for each of the $p(p-1)$ index pairs and each of the $N_q(T,a)^2$ tuple pairs, contributing $1-p^{-1}$.

\noindent(ii) $i_1=i_2 \ne j_1=j_2$. The summand factorizes as
\[
\E\left[\prod_{k=1}^{a}x_{i_1,t_k}x_{i_1,u_k}\right]\E\left[\prod_{k=1}^{a}x_{j_1,t_k}x_{j_1,u_k}\right]=O(1),
\]
and is non-zero only if every $t \in S(\bt_a)$ has a partner $u \in S(\bu_a)$ with $|t-u|\le1$. For each $\bt_a$ there are at most $3^{a}$ such $\bu_a$, so only $O(T^{a})$ of the $N_q(T,a)^2$ tuple pairs contribute. With $O(p^2)$ index pairs, this case contributes $O(T^{a}/T^{2a})=O(T^{-a})$.

\noindent(iii) $i_1=j_2 \ne i_2=j_1$. Identical to case (ii), contributing $O(T^{-a})$.

\noindent(iv) $i_1=j_1=i_2=j_2$. Here the summand is $O(1)$ for all $N_q(T,a)^2$ tuple pairs, but only $p$ index choices remain, so the contribution is $O(p^{-1})$.

Summing the four cases, $\E[r(a)^2]=1+O(p^{-1})+O(T^{-a})=1+o(1)$, whence $\Var[r(a)]=\E[r(a)^2]-1\to0$ and $r(a)\pto1$. Combined with $U_q(a)/\sigma(a)\pto\infty$, this yields $U_q(a)/\hat{\sigma}(a)\pto\infty$.
\end{proof}


\section{Conclusions}\label{sec:conc}

We have developed a family of high-dimensional white noise tests that dispense with cross-sectional independence. The even order $a$ interpolates between sum-type and maximum-type statistics, and asymptotic independence across orders allows several of them to be aggregated. Among the equal-weight, Fisher, and Cauchy combinations, $U_{\rm adp}$ controls size most accurately and, at the cost of a small loss of power in the most favourable scenarios, stays close to the best single order under sparse, intermediate, and dense alternatives; it is our default recommendation. Its size is accurate under Gaussian, heavy-tailed, and skewed innovations alike. Theoretically, linking cross-sectional correlation to a graph structure is what makes the derivation tractable.

The simulations also show that the levels of $M_q$, $S_q$, and $F_q$ can distort severely under second-order uncorrelated nonlinear structures, whereas $U_2$ and $U_4$ remain close to nominal under conditional heteroscedasticity and $U_{\rm adp}$ inflates only moderately; when nonlinear dependence is suspected, $U_2$ or $U_4$ is preferable, cross-checked against $U_{\rm adp}$. A further extension concerns the power theory: explicit noncentrality parameters under local alternatives, as discussed in Remark~\ref{rem:power}, would permit local power comparisons across orders and a data-driven choice of $a$.


\section*{Acknowledgements}

The author would like to thank her advisor, Prof. Weiming Li, for his guidance and support. The author is also grateful to the editor, an associate editor, and two anonymous referees for their constructive comments.
\par


\section*{Appendix}

This appendix collects the simulation results omitted from the main text. All experiments are designed as in Section~\ref{sec:sim} of the main text: the lag order is $q=1$, the nominal significance level is $5\%$, and $U_a$ denotes the studentized statistic $U_q(a)/\hat{\sigma}(a)$; the combination tests $U_{\rm adp}$, $U_{\rm F}$, and $U_{\rm C}$, as well as the competitors $M_q$, $S_q$, and $F_q$, are as defined there. Tables~\ref{tab:size_b} and~\ref{tab:size_c} report the empirical sizes under models (b) and (c), in which $z_{it}\sim t(8)/\sqrt{4/3}$ and $z_{it}\sim\mathrm{Gamma}(4,0.5)-2$, respectively; Tables~\ref{tab:pow_vma_sparse}--\ref{tab:pow_vma_dense} report the empirical power under the VMA(1) alternative with sparse, intermediate, and dense signals, respectively; and Table~\ref{tab:wn_ii} reports the empirical sizes under weak white noise with the AR(1)-decaying covariance structure $\bSig_0=\bigl(0.5^{|j_1-j_2|}\bigr)_{1 \le j_1,j_2 \le p}$. In Tables~\ref{tab:size_b}--\ref{tab:pow_vma_dense}, the upper and lower panels correspond to $\bSig_0=\bI_p$ and $\bSig_0=\bigl(0.5^{|j_1-j_2|}\bigr)$, respectively. All results are based on 2,000 Monte Carlo replications, except for the GARCH$(0.1,0.85)$ case in Table~\ref{tab:wn_ii}, which is based on 5,000 replications.

\begin{table}[t!]
\centering
\caption{Empirical sizes (in percentage) of the tests under model (b), in which $z_{it} \sim t(8)/\sqrt{4/3}$, at the $5\%$ nominal level with $q=1$. The upper and lower panels correspond to $\bSig_0=\bI_p$ and $\bSig_0=\bigl(0.5^{|j_1-j_2|}\bigr)_{1 \le j_1,j_2 \le p}$, respectively. Results are based on 2,000 Monte Carlo replications.}
\label{tab:size_b}
\begin{tabular}{rrrrrrrrrrrr} \hline
\multicolumn{1}{c}{$p$} & \multicolumn{1}{c}{$T$} & \multicolumn{1}{c}{$p/T$} & \multicolumn{1}{c}{$U_2$} & \multicolumn{1}{c}{$U_4$} & \multicolumn{1}{c}{$U_6$} & \multicolumn{1}{c}{$U_{\rm adp}$} & \multicolumn{1}{c}{$U_{\rm F}$} & \multicolumn{1}{c}{$U_{\rm C}$} & \multicolumn{1}{c}{$M_q$} & \multicolumn{1}{c}{$S_q$} & \multicolumn{1}{c}{$F_q$} \\ \hline
\multicolumn{12}{c}{$\bSig_0=\bI_p$}\\ \hline
50 & 100 & 0.5 & 5.00 & 5.65 & 4.10 & 5.70 & 5.75 & 5.40 & 4.65 & 4.85 & 3.45 \\
100 & 200 & 0.5 & 5.40 & 6.15 & 5.35 & 5.50 & 5.75 & 6.50 & 6.00 & 5.15 & 4.20 \\
200 & 400 & 0.5 & 5.30 & 5.45 & 4.80 & 4.95 & 5.10 & 5.10 & 4.90 & 5.15 & 3.95 \\
400 & 800 & 0.5 & 5.25 & 5.25 & 5.75 & 5.50 & 6.00 & 5.65 & 5.50 & 5.20 & 5.35 \\
100 & 100 & 1.0 & 5.65 & 4.85 & 4.45 & 5.80 & 6.10 & 5.35 & 7.35 & 5.05 & 3.50 \\
200 & 200 & 1.0 & 4.15 & 5.20 & 4.45 & 4.55 & 5.00 & 5.20 & 5.70 & 4.20 & 3.05 \\
400 & 400 & 1.0 & 5.30 & 4.50 & 6.45 & 5.25 & 5.80 & 6.00 & 4.80 & 5.05 & 4.45 \\
800 & 800 & 1.0 & 5.45 & 5.35 & 6.00 & 5.60 & 6.05 & 6.30 & 5.65 & 5.45 & 5.60 \\
150 & 100 & 1.5 & 5.60 & 4.30 & 5.85 & 5.65 & 6.10 & 6.00 & 7.05 & 5.30 & 3.30 \\
300 & 200 & 1.5 & 4.75 & 4.85 & 4.95 & 4.75 & 5.10 & 5.20 & 6.55 & 4.55 & 3.65 \\
600 & 400 & 1.5 & 4.40 & 4.85 & 5.50 & 4.60 & 5.30 & 5.10 & 5.60 & 4.45 & 3.50 \\
1200 & 800 & 1.5 & 5.10 & 5.10 & 4.70 & 4.40 & 5.35 & 5.50 & 4.70 & 4.85 & 4.80 \\ \hline
\multicolumn{12}{c}{$\bSig_0=\bigl(0.5^{|j_1-j_2|}\bigr)_{1 \le j_1,j_2 \le p}$}\\ \hline
50 & 100 & 0.5 & 5.05 & 6.40 & 4.10 & 6.55 & 6.75 & 6.00 & 5.35 & 5.15 & 4.00 \\
100 & 200 & 0.5 & 5.55 & 5.65 & 4.65 & 5.80 & 6.15 & 5.85 & 5.90 & 5.20 & 4.95 \\
200 & 400 & 0.5 & 4.35 & 5.65 & 5.65 & 5.55 & 6.20 & 6.00 & 6.05 & 4.20 & 4.45 \\
400 & 800 & 0.5 & 5.15 & 4.75 & 6.00 & 4.65 & 4.95 & 6.05 & 4.10 & 4.80 & 4.45 \\
100 & 100 & 1.0 & 4.80 & 4.75 & 5.10 & 5.25 & 5.60 & 6.15 & 6.45 & 4.55 & 3.20 \\
200 & 200 & 1.0 & 5.70 & 5.05 & 4.80 & 5.30 & 5.50 & 5.25 & 4.70 & 5.70 & 3.90 \\
400 & 400 & 1.0 & 4.65 & 5.75 & 5.30 & 5.60 & 5.15 & 5.00 & 4.35 & 4.75 & 3.40 \\
800 & 800 & 1.0 & 5.15 & 4.50 & 5.10 & 5.05 & 5.45 & 5.00 & 4.75 & 5.05 & 4.10 \\
150 & 100 & 1.5 & 5.10 & 5.20 & 5.00 & 5.65 & 4.95 & 4.60 & 4.40 & 4.95 & 2.50 \\
300 & 200 & 1.5 & 4.45 & 4.75 & 5.70 & 5.30 & 5.10 & 4.90 & 5.80 & 4.40 & 2.65 \\
600 & 400 & 1.5 & 5.40 & 5.35 & 5.80 & 5.45 & 5.80 & 6.45 & 5.00 & 5.50 & 4.35 \\
1200 & 800 & 1.5 & 4.65 & 4.85 & 5.30 & 5.25 & 4.95 & 5.00 & 5.40 & 4.45 & 4.05 \\ \hline
\end{tabular}
\end{table}

\begin{table}[t!]
\centering
\caption{Empirical sizes (in percentage) of the tests under model (c), in which $z_{it}\sim\mathrm{Gamma}(4,0.5)-2$, at the $5\%$ nominal level with $q=1$. The upper and lower panels correspond to $\bSig_0=\bI_p$ and $\bSig_0=\bigl(0.5^{|j_1-j_2|}\bigr)_{1 \le j_1,j_2 \le p}$, respectively. Results are based on 2,000 Monte Carlo replications.}
\label{tab:size_c}
\begin{tabular}{rrrrrrrrrrrr} \hline
\multicolumn{1}{c}{$p$} & \multicolumn{1}{c}{$T$} & \multicolumn{1}{c}{$p/T$} & \multicolumn{1}{c}{$U_2$} & \multicolumn{1}{c}{$U_4$} & \multicolumn{1}{c}{$U_6$} & \multicolumn{1}{c}{$U_{\rm adp}$} & \multicolumn{1}{c}{$U_{\rm F}$} & \multicolumn{1}{c}{$U_{\rm C}$} & \multicolumn{1}{c}{$M_q$} & \multicolumn{1}{c}{$S_q$} & \multicolumn{1}{c}{$F_q$} \\ \hline
\multicolumn{12}{c}{$\bSig_0=\bI_p$}\\ \hline
50 & 100 & 0.5 & 4.85 & 5.60 & 4.90 & 6.00 & 6.00 & 6.75 & 7.35 & 4.90 & 4.70 \\
100 & 200 & 0.5 & 4.05 & 4.85 & 5.90 & 5.95 & 5.85 & 6.15 & 6.10 & 4.35 & 4.45 \\
200 & 400 & 0.5 & 4.55 & 5.35 & 4.30 & 5.00 & 5.60 & 4.85 & 7.65 & 4.65 & 4.70 \\
400 & 800 & 0.5 & 4.15 & 4.90 & 6.00 & 5.25 & 4.95 & 4.95 & 8.05 & 4.20 & 4.05 \\
100 & 100 & 1.0 & 5.30 & 5.10 & 5.95 & 5.70 & 6.05 & 6.45 & 9.10 & 5.05 & 4.00 \\
200 & 200 & 1.0 & 4.55 & 5.90 & 4.75 & 5.45 & 5.40 & 5.25 & 8.15 & 4.85 & 3.90 \\
400 & 400 & 1.0 & 5.25 & 4.60 & 5.25 & 5.55 & 5.95 & 6.30 & 6.55 & 5.60 & 5.05 \\
800 & 800 & 1.0 & 5.60 & 4.50 & 4.80 & 5.60 & 5.40 & 5.40 & 8.95 & 5.85 & 5.60 \\
150 & 100 & 1.5 & 4.60 & 5.65 & 4.85 & 5.40 & 6.10 & 5.55 & 9.55 & 4.25 & 2.90 \\
300 & 200 & 1.5 & 4.80 & 5.40 & 5.50 & 5.70 & 5.85 & 5.45 & 8.60 & 4.60 & 4.55 \\
600 & 400 & 1.5 & 4.60 & 5.05 & 4.85 & 5.30 & 5.55 & 5.20 & 9.00 & 4.00 & 4.70 \\
1200 & 800 & 1.5 & 4.90 & 4.75 & 5.65 & 4.80 & 4.95 & 5.10 & 6.15 & 5.00 & 5.60 \\ \hline
\multicolumn{12}{c}{$\bSig_0=\bigl(0.5^{|j_1-j_2|}\bigr)_{1 \le j_1,j_2 \le p}$}\\ \hline
50 & 100 & 0.5 & 5.60 & 4.80 & 3.65 & 5.40 & 6.70 & 6.05 & 5.00 & 5.65 & 4.05 \\
100 & 200 & 0.5 & 6.10 & 5.45 & 4.75 & 6.05 & 6.10 & 6.35 & 4.40 & 5.75 & 5.10 \\
200 & 400 & 0.5 & 5.10 & 4.50 & 5.00 & 5.10 & 5.30 & 4.80 & 5.45 & 4.75 & 4.25 \\
400 & 800 & 0.5 & 5.50 & 5.20 & 6.20 & 5.40 & 6.35 & 5.70 & 3.40 & 5.15 & 4.55 \\
100 & 100 & 1.0 & 5.50 & 5.25 & 4.75 & 5.70 & 5.90 & 5.75 & 5.50 & 5.65 & 3.40 \\
200 & 200 & 1.0 & 4.55 & 5.15 & 4.90 & 5.20 & 5.30 & 5.80 & 7.00 & 5.15 & 3.85 \\
400 & 400 & 1.0 & 4.70 & 5.15 & 6.30 & 6.10 & 5.85 & 5.75 & 6.35 & 4.70 & 4.00 \\
800 & 800 & 1.0 & 5.50 & 6.40 & 4.95 & 5.05 & 5.15 & 5.55 & 4.00 & 5.40 & 4.25 \\
150 & 100 & 1.5 & 5.10 & 5.05 & 4.65 & 6.15 & 5.85 & 5.70 & 7.20 & 5.35 & 3.55 \\
300 & 200 & 1.5 & 5.25 & 5.55 & 5.45 & 6.25 & 6.30 & 5.80 & 7.25 & 5.35 & 3.65 \\
600 & 400 & 1.5 & 5.65 & 4.95 & 6.20 & 5.85 & 6.15 & 6.40 & 5.45 & 5.50 & 4.15 \\
1200 & 800 & 1.5 & 4.90 & 4.85 & 4.95 & 4.40 & 4.65 & 5.85 & 4.90 & 5.10 & 4.55 \\ \hline
\end{tabular}
\end{table}

\begin{table}[t!]
\centering
\caption{Empirical power (in percentage) of the tests under the VMA(1) alternative with sparse signals ($d=\lfloor0.05p\rfloor$) at the $5\%$ nominal level with $q=1$ and $\theta=0.2$. The upper and lower panels correspond to $\bSig_0=\bI_p$ and $\bSig_0=\bigl(0.5^{|j_1-j_2|}\bigr)_{1 \le j_1,j_2 \le p}$, respectively. Results are based on 2,000 Monte Carlo replications.}
\label{tab:pow_vma_sparse}
\begin{tabular}{rrrrrrrrrrrr} \hline
\multicolumn{1}{c}{$p$} & \multicolumn{1}{c}{$T$} & \multicolumn{1}{c}{$p/T$} & \multicolumn{1}{c}{$U_2$} & \multicolumn{1}{c}{$U_4$} & \multicolumn{1}{c}{$U_6$} & \multicolumn{1}{c}{$U_{\rm adp}$} & \multicolumn{1}{c}{$U_{\rm F}$} & \multicolumn{1}{c}{$U_{\rm C}$} & \multicolumn{1}{c}{$M_q$} & \multicolumn{1}{c}{$S_q$} & \multicolumn{1}{c}{$F_q$} \\ \hline
\multicolumn{12}{c}{$\bSig_0=\bI_p$}\\ \hline
50 & 100 & 0.5 & 6.55 & 8.10 & 6.40 & 10.00 & 10.00 & 9.15 & 6.00 & 6.65 & 4.80 \\
100 & 200 & 0.5 & 9.25 & 17.65 & 19.60 & 24.40 & 24.20 & 22.85 & 13.20 & 9.35 & 11.15 \\
200 & 400 & 0.5 & 13.30 & 72.80 & 90.35 & 92.00 & 92.05 & 90.90 & 77.70 & 14.10 & 68.10 \\
100 & 100 & 1.0 & 5.90 & 7.05 & 7.60 & 9.95 & 9.45 & 8.35 & 5.85 & 6.20 & 3.50 \\
200 & 200 & 1.0 & 9.70 & 17.50 & 22.05 & 27.60 & 27.65 & 25.00 & 19.30 & 9.90 & 12.40 \\
400 & 400 & 1.0 & 12.80 & 76.60 & 96.70 & 97.00 & 97.35 & 96.65 & 82.45 & 14.40 & 72.80 \\
150 & 100 & 1.5 & 6.65 & 7.65 & 7.15 & 8.60 & 9.45 & 9.20 & 5.95 & 7.05 & 3.30 \\
300 & 200 & 1.5 & 8.35 & 17.05 & 22.75 & 26.25 & 26.75 & 24.55 & 16.95 & 9.15 & 9.65 \\
600 & 400 & 1.5 & 13.65 & 77.15 & 98.65 & 98.60 & 99.20 & 98.85 & 86.15 & 15.70 & 76.05 \\ \hline
\multicolumn{12}{c}{$\bSig_0=\bigl(0.5^{|j_1-j_2|}\bigr)_{1 \le j_1,j_2 \le p}$}\\ \hline
50 & 100 & 0.5 & 6.35 & 7.10 & 5.55 & 9.20 & 8.85 & 8.05 & 5.85 & 6.10 & 4.20 \\
100 & 200 & 0.5 & 8.45 & 17.55 & 19.25 & 24.10 & 23.85 & 21.85 & 12.55 & 8.35 & 11.45 \\
200 & 400 & 0.5 & 14.75 & 68.45 & 87.05 & 88.95 & 89.45 & 87.45 & 74.25 & 15.70 & 65.70 \\
100 & 100 & 1.0 & 6.35 & 7.65 & 6.80 & 9.45 & 8.95 & 7.90 & 5.55 & 5.85 & 3.80 \\
200 & 200 & 1.0 & 9.15 & 17.85 & 21.90 & 26.30 & 25.95 & 23.90 & 11.80 & 9.40 & 11.20 \\
400 & 400 & 1.0 & 14.65 & 73.20 & 94.35 & 94.45 & 94.90 & 94.00 & 82.30 & 14.50 & 69.00 \\
150 & 100 & 1.5 & 6.40 & 8.00 & 7.70 & 9.45 & 9.30 & 9.45 & 5.90 & 6.45 & 4.20 \\
300 & 200 & 1.5 & 8.70 & 15.95 & 20.10 & 24.60 & 25.40 & 22.50 & 14.10 & 9.65 & 9.00 \\
600 & 400 & 1.5 & 14.50 & 74.80 & 97.05 & 97.10 & 97.85 & 97.35 & 84.45 & 15.10 & 71.80 \\ \hline
\end{tabular}
\end{table}

\begin{table}[t!]
\centering
\caption{Empirical power (in percentage) of the tests under the VMA(1) alternative with intermediate signals ($d=\lfloor0.30p\rfloor$) at the $5\%$ nominal level with $q=1$ and $\theta=0.2$. The upper and lower panels correspond to $\bSig_0=\bI_p$ and $\bSig_0=\bigl(0.5^{|j_1-j_2|}\bigr)_{1 \le j_1,j_2 \le p}$, respectively. Results are based on 2,000 Monte Carlo replications.}
\label{tab:pow_vma_medium}
\begin{tabular}{rrrrrrrrrrrr} \hline
\multicolumn{1}{c}{$p$} & \multicolumn{1}{c}{$T$} & \multicolumn{1}{c}{$p/T$} & \multicolumn{1}{c}{$U_2$} & \multicolumn{1}{c}{$U_4$} & \multicolumn{1}{c}{$U_6$} & \multicolumn{1}{c}{$U_{\rm adp}$} & \multicolumn{1}{c}{$U_{\rm F}$} & \multicolumn{1}{c}{$U_{\rm C}$} & \multicolumn{1}{c}{$M_q$} & \multicolumn{1}{c}{$S_q$} & \multicolumn{1}{c}{$F_q$} \\ \hline
\multicolumn{12}{c}{$\bSig_0=\bI_p$}\\ \hline
50 & 100 & 0.5 & 19.70 & 24.60 & 16.00 & 29.55 & 30.60 & 28.40 & 10.85 & 26.85 & 20.70 \\
100 & 200 & 0.5 & 49.90 & 88.75 & 78.95 & 95.30 & 95.15 & 93.50 & 59.25 & 68.35 & 78.75 \\
200 & 400 & 0.5 & 94.50 & 100.00 & 100.00 & 100.00 & 100.00 & 100.00 & 100.00 & 99.35 & 100.00 \\
100 & 100 & 1.0 & 21.40 & 22.75 & 16.05 & 33.05 & 32.85 & 29.00 & 10.45 & 35.35 & 25.40 \\
200 & 200 & 1.0 & 50.05 & 91.75 & 87.10 & 98.25 & 98.40 & 97.20 & 51.95 & 82.20 & 85.80 \\
400 & 400 & 1.0 & 95.00 & 100.00 & 100.00 & 100.00 & 100.00 & 100.00 & 100.00 & 100.00 & 100.00 \\
150 & 100 & 1.5 & 19.90 & 21.45 & 14.85 & 32.55 & 31.55 & 27.75 & 5.45 & 44.50 & 30.60 \\
300 & 200 & 1.5 & 49.80 & 93.70 & 90.75 & 99.35 & 99.25 & 98.15 & 49.60 & 92.45 & 93.10 \\
600 & 400 & 1.5 & 94.50 & 100.00 & 100.00 & 100.00 & 100.00 & 100.00 & 100.00 & 100.00 & 100.00 \\ \hline
\multicolumn{12}{c}{$\bSig_0=\bigl(0.5^{|j_1-j_2|}\bigr)_{1 \le j_1,j_2 \le p}$}\\ \hline
50 & 100 & 0.5 & 20.45 & 22.85 & 15.30 & 29.80 & 31.35 & 29.15 & 10.80 & 25.75 & 21.45 \\
100 & 200 & 0.5 & 49.45 & 82.65 & 72.40 & 90.75 & 89.95 & 87.40 & 49.25 & 59.25 & 70.00 \\
200 & 400 & 0.5 & 93.25 & 100.00 & 100.00 & 100.00 & 100.00 & 100.00 & 100.00 & 97.95 & 100.00 \\
100 & 100 & 1.0 & 22.30 & 24.20 & 16.00 & 33.20 & 33.70 & 29.60 & 12.10 & 30.90 & 24.70 \\
200 & 200 & 1.0 & 46.60 & 89.20 & 81.10 & 95.65 & 95.25 & 93.35 & 45.90 & 69.30 & 73.95 \\
400 & 400 & 1.0 & 93.50 & 100.00 & 100.00 & 100.00 & 100.00 & 100.00 & 99.95 & 99.65 & 100.00 \\
150 & 100 & 1.5 & 20.95 & 24.20 & 15.25 & 32.20 & 31.55 & 29.70 & 8.70 & 34.20 & 24.00 \\
300 & 200 & 1.5 & 50.70 & 91.80 & 89.00 & 98.10 & 97.85 & 96.00 & 49.45 & 78.60 & 81.95 \\
600 & 400 & 1.5 & 95.00 & 100.00 & 100.00 & 100.00 & 100.00 & 100.00 & 100.00 & 99.95 & 100.00 \\ \hline
\end{tabular}
\end{table}

\begin{table}[t!]
\centering
\caption{Empirical power (in percentage) of the tests under the VMA(1) alternative with dense signals ($d=\lfloor0.95p\rfloor$) at the $5\%$ nominal level with $q=1$ and $\theta=0.2$. The upper and lower panels correspond to $\bSig_0=\bI_p$ and $\bSig_0=\bigl(0.5^{|j_1-j_2|}\bigr)_{1 \le j_1,j_2 \le p}$, respectively. Results are based on 2,000 Monte Carlo replications.}
\label{tab:pow_vma_dense}
\begin{tabular}{rrrrrrrrrrrr} \hline
\multicolumn{1}{c}{$p$} & \multicolumn{1}{c}{$T$} & \multicolumn{1}{c}{$p/T$} & \multicolumn{1}{c}{$U_2$} & \multicolumn{1}{c}{$U_4$} & \multicolumn{1}{c}{$U_6$} & \multicolumn{1}{c}{$U_{\rm adp}$} & \multicolumn{1}{c}{$U_{\rm F}$} & \multicolumn{1}{c}{$U_{\rm C}$} & \multicolumn{1}{c}{$M_q$} & \multicolumn{1}{c}{$S_q$} & \multicolumn{1}{c}{$F_q$} \\ \hline
\multicolumn{12}{c}{$\bSig_0=\bI_p$}\\ \hline
50 & 100 & 0.5 & 76.95 & 67.00 & 34.10 & 82.90 & 86.60 & 83.80 & 29.80 & 99.35 & 98.65 \\
100 & 200 & 0.5 & 99.70 & 100.00 & 99.80 & 100.00 & 100.00 & 100.00 & 91.00 & 100.00 & 100.00 \\
200 & 400 & 0.5 & 100.00 & 100.00 & 100.00 & 100.00 & 100.00 & 100.00 & 100.00 & 100.00 & 100.00 \\
100 & 100 & 1.0 & 76.20 & 69.20 & 36.35 & 86.05 & 88.55 & 86.35 & 19.00 & 100.00 & 100.00 \\
200 & 200 & 1.0 & 99.70 & 100.00 & 99.90 & 100.00 & 100.00 & 100.00 & 90.65 & 100.00 & 100.00 \\
400 & 400 & 1.0 & 100.00 & 100.00 & 100.00 & 100.00 & 100.00 & 100.00 & 100.00 & 100.00 & 100.00 \\
150 & 100 & 1.5 & 75.15 & 68.85 & 36.65 & 88.40 & 89.65 & 86.00 & 20.55 & 100.00 & 100.00 \\
300 & 200 & 1.5 & 99.95 & 100.00 & 100.00 & 100.00 & 100.00 & 100.00 & 83.95 & 100.00 & 100.00 \\
600 & 400 & 1.5 & 100.00 & 100.00 & 100.00 & 100.00 & 100.00 & 100.00 & 100.00 & 100.00 & 100.00 \\ \hline
\multicolumn{12}{c}{$\bSig_0=\bigl(0.5^{|j_1-j_2|}\bigr)_{1 \le j_1,j_2 \le p}$}\\ \hline
50 & 100 & 0.5 & 73.25 & 59.70 & 30.80 & 75.05 & 80.10 & 77.60 & 21.05 & 94.50 & 91.95 \\
100 & 200 & 0.5 & 99.75 & 100.00 & 99.20 & 100.00 & 100.00 & 100.00 & 86.20 & 100.00 & 100.00 \\
200 & 400 & 0.5 & 100.00 & 100.00 & 100.00 & 100.00 & 100.00 & 100.00 & 100.00 & 100.00 & 100.00 \\
100 & 100 & 1.0 & 76.25 & 64.90 & 35.35 & 82.05 & 85.65 & 83.15 & 23.85 & 99.70 & 99.40 \\
200 & 200 & 1.0 & 99.70 & 100.00 & 100.00 & 100.00 & 100.00 & 100.00 & 85.80 & 100.00 & 100.00 \\
400 & 400 & 1.0 & 100.00 & 100.00 & 100.00 & 100.00 & 100.00 & 100.00 & 100.00 & 100.00 & 100.00 \\
150 & 100 & 1.5 & 76.80 & 67.65 & 37.50 & 87.00 & 89.20 & 85.75 & 22.10 & 100.00 & 100.00 \\
300 & 200 & 1.5 & 99.70 & 100.00 & 100.00 & 100.00 & 100.00 & 100.00 & 87.55 & 100.00 & 100.00 \\
600 & 400 & 1.5 & 100.00 & 100.00 & 100.00 & 100.00 & 100.00 & 100.00 & 100.00 & 100.00 & 100.00 \\ \hline
\end{tabular}
\end{table}

\begin{table}[t!]
\centering
\caption{Empirical sizes (in percentage) of the tests under weak white noise with $\bSig_0=\bigl(0.5^{|j_1-j_2|}\bigr)_{1 \le j_1,j_2 \le p}$ at the $5\%$ nominal level.}
\label{tab:wn_ii}
\begin{tabular}{rrrrrrrrrrrr} \hline
\multicolumn{1}{c}{$p$} & \multicolumn{1}{c}{$T$} & \multicolumn{1}{c}{$p/T$} & \multicolumn{1}{c}{$U_2$} & \multicolumn{1}{c}{$U_4$} & \multicolumn{1}{c}{$U_6$} & \multicolumn{1}{c}{$U_{\rm adp}$} & \multicolumn{1}{c}{$U_{\rm F}$} & \multicolumn{1}{c}{$U_{\rm C}$} & \multicolumn{1}{c}{$M_q$} & \multicolumn{1}{c}{$S_q$} & \multicolumn{1}{c}{$F_q$} \\ \hline
\multicolumn{12}{c}{ARCH(1)}\\ \hline
200 & 200 & 1.0 & 5.58 & 6.34 & 5.36 & 6.44 & 6.66 & 6.28 & 7.34 & 10.62 & 6.02 \\
400 & 400 & 1.0 & 5.22 & 5.58 & 5.66 & 5.94 & 6.26 & 6.22 & 9.12 & 9.80 & 7.08 \\
800 & 800 & 1.0 & 4.50 & 4.70 & 5.25 & 5.25 & 5.65 & 5.40 & 9.00 & 8.40 & 7.60 \\
300 & 200 & 1.5 & 5.50 & 5.50 & 6.25 & 5.80 & 5.50 & 6.45 & 10.70 & 10.40 & 5.50 \\
600 & 400 & 1.5 & 4.65 & 5.45 & 5.85 & 5.45 & 5.90 & 6.15 & 10.90 & 8.75 & 6.55 \\
1200 & 800 & 1.5 & 4.95 & 6.50 & 5.90 & 5.80 & 6.40 & 6.05 & 10.35 & 10.05 & 7.80 \\ \hline
\multicolumn{12}{c}{GARCH$(0.1,0.6)$}\\ \hline
200 & 200 & 1.0 & 4.65 & 5.00 & 6.40 & 6.05 & 6.40 & 6.65 & 5.15 & 5.95 & 3.60 \\
400 & 400 & 1.0 & 4.55 & 4.85 & 5.90 & 5.40 & 5.25 & 5.40 & 4.50 & 5.35 & 3.90 \\
800 & 800 & 1.0 & 5.25 & 5.00 & 6.00 & 5.65 & 5.45 & 6.10 & 4.85 & 6.20 & 4.35 \\
300 & 200 & 1.5 & 4.80 & 4.80 & 4.85 & 5.05 & 5.25 & 5.35 & 5.65 & 6.00 & 3.20 \\
600 & 400 & 1.5 & 4.85 & 4.20 & 5.85 & 5.45 & 5.50 & 5.55 & 6.45 & 6.15 & 3.65 \\
1200 & 800 & 1.5 & 5.45 & 5.10 & 6.65 & 5.95 & 6.25 & 6.55 & 4.65 & 6.95 & 4.45 \\ \hline
\multicolumn{12}{c}{GARCH$(0.1,0.85)$}\\ \hline
200 & 200 & 1.0 & 4.94 & 6.44 & 7.64 & 6.64 & 7.98 & 7.96 & 4.96 & 7.04 & 3.78 \\
400 & 400 & 1.0 & 5.28 & 5.78 & 9.38 & 7.30 & 8.36 & 8.84 & 5.90 & 7.42 & 4.82 \\
800 & 800 & 1.0 & 4.62 & 5.16 & 7.20 & 5.96 & 6.82 & 7.00 & 7.56 & 6.90 & 5.28 \\
300 & 200 & 1.5 & 5.38 & 6.18 & 8.38 & 6.78 & 8.16 & 8.46 & 5.92 & 7.38 & 4.12 \\
600 & 400 & 1.5 & 4.90 & 6.04 & 8.72 & 6.54 & 8.40 & 8.18 & 6.38 & 7.10 & 4.44 \\
1200 & 800 & 1.5 & 4.98 & 5.54 & 7.22 & 6.98 & 7.26 & 7.00 & 7.92 & 7.32 & 4.92 \\ \hline
\multicolumn{12}{c}{Bilinear}\\ \hline
200 & 200 & 1.0 & 5.95 & 5.90 & 6.25 & 7.15 & 7.50 & 7.75 & 10.70 & 8.80 & 6.60 \\
400 & 400 & 1.0 & 6.80 & 6.45 & 5.50 & 7.40 & 7.10 & 6.60 & 8.30 & 9.70 & 7.15 \\
800 & 800 & 1.0 & 8.10 & 5.30 & 5.95 & 7.40 & 8.00 & 7.00 & 7.65 & 11.20 & 7.45 \\
300 & 200 & 1.5 & 6.85 & 5.65 & 5.35 & 7.80 & 7.25 & 7.20 & 8.80 & 9.75 & 7.35 \\
600 & 400 & 1.5 & 6.35 & 5.20 & 5.70 & 6.00 & 6.00 & 6.50 & 9.40 & 9.10 & 6.70 \\
1200 & 800 & 1.5 & 7.00 & 6.00 & 5.35 & 7.20 & 6.90 & 6.25 & 9.95 & 10.25 & 8.10 \\ \hline
\end{tabular}
\end{table}


\bibliographystyle{elsarticle-harv}
\bibliography{ref}

\end{document}